\documentclass[letterpaper,twocolumn,10pt]{article}

\usepackage{usenix,epsfig,endnotes}

\usepackage{balance}
\usepackage{amsopn}
\usepackage{amsmath} 
\usepackage{amsthm}
\usepackage{url}
\usepackage{multirow}
\usepackage{hyperref}
\usepackage{cite}
\usepackage{cleveref}

\newtheorem{theorem}{Theorem}
\newtheorem{corollary}{Corollary}
\newtheorem{lemma}{Lemma}
\theoremstyle{definition}
\newtheorem{definition}{Definition}[section]

\usepackage{xcolor}
\usepackage{paralist}
\usepackage{graphicx}
\mathchardef\mhyphen="2D

\newcommand{\Sim}[1]{\mathsf{Sim}_{#1}}

\newcommand{\ourscheme}{Pyramid ORAM}
\newcommand{\ourhashtable}{Zigzag Hash Table}
\newcommand{\ourhashtableS}{ZHT}
\newcommand{\route}{\mathsf{route}}
\newcommand{\build}{\mathsf{build}}
\newcommand{\search}{\mathsf{search}}
\newcommand{\throw}{\mathsf{throw}}
\newcommand{\real}{\mathsf{real}}
\newcommand{\true}{\mathsf{true}}
\newcommand{\false}{\mathsf{false}}

\newcommand{\inc}{\mathsf{in}}
\newcommand{\out}{\mathsf{out}}

\newcommand{\N}{N}
\newcommand{\D}{D}

\newcommand{\secsize}{\mathsf{s}}

\newcommand{\p}{\mathsf{p}} 

\newcommand{\n}{n}

\newcommand{\BSize}{c}
\newcommand{\NTable}{k}

\newcommand{\key}{\mathsf{key}}
\newcommand{\val}{\mathsf{val}}

\newcommand{\firstLevelSize}{p}
\newcommand{\dummy}{\mathsf{dummy}}

\newcommand{\Level}{L}

\makeatletter
\newcommand{\Rmnum}[1]{\expandafter\@slowromancap\romannumeral #1@}
\makeatother

\title{The Pyramid Scheme: Oblivious RAM for Trusted Processors}

\author{
  Manuel Costa
  \and \makebox[.2\linewidth]{Lawrence Esswood\thanks{Work done at Microsoft Research.}}\\The University of Cambridge\\
  \and \makebox[.3\linewidth]{Olga Ohrimenko}\\Microsoft Research\\
  \and \makebox[.3\linewidth]{Felix Schuster}\\Microsoft Research\\
  \and \makebox[.2\linewidth]{Sameer Wagh$^*$}\\Princeton University\\
}

\author{
  {\rm Manuel Costa, Lawrence Esswood\thanks{Work done at Microsoft Research; affiliated with University of Cambridge, UK.}, Olga Ohrimenko, Felix Schuster, Sameer Wagh\thanks{Work done at Microsoft Research; affiliated with Princeton University, USA.}}\\ 
  Microsoft Research
  }

\begin{document}

\maketitle

\begin{abstract}
Modern processors, e.g., Intel SGX, allow applications to isolate
secret code and data in encrypted memory regions called enclaves.
While encryption effectively hides the contents of memory, the sequence of
address references issued by the secret code leaks information. This is a serious
problem because these leaks can easily break
the confidentiality guarantees of enclaves.

In this paper, we explore Oblivious RAM (ORAM) designs that prevent these
information leaks under the constraints of modern SGX processors.
Most ORAMs are a poor fit for these processors
because they have high constant overhead factors or require large private memories,
which are not available in these processors.
We address these limitations with a new hierarchical ORAM construction, the Pyramid ORAM,
that is optimized towards online bandwidth cost and small blocks.
It uses a new hashing scheme that circumvents
the complexity of previous hierarchical schemes.

We present an efficient x64-optimized implementation of Pyramid ORAM that uses only
the processor's registers as private memory.
We compare \ourscheme{}
with Circuit ORAM,
a state-of-the-art tree-based ORAM scheme that also uses constant private memory.
Pyramid ORAM has better
online asymptotical complexity than Circuit ORAM.
Our implementation of \ourscheme{} and Circuit ORAM validates
this: as all hierarchical schemes, Pyramid ORAM has high variance of access latencies;
although latency can be high for some accesses,
for typical configurations Pyramid ORAM provides access latencies that are
{8X} better than Circuit ORAM for {99\%} of accesses.
Although the best known hierarchical ORAM has better asymptotical complexity,
Pyramid ORAM has significantly lower constant overhead factors, making it the preferred choice
in practice.
\end{abstract}

\section{Introduction} \label{sec:introduction} 

Intel SGX~\cite{sgx3} allows applications to isolate secret code and data in encrypted memory regions called enclaves. Enclaves have many interesting applications, for instance
they can protect cloud workloads~\cite{haven,vc3,scone,ryoan} in malicious or compromised environments, or guard secrets
in client devices~\cite{sgx3}.

While encryption effectively hides the contents of memory, the sequence of memory address references issued by enclave code leaks information. If the attacker has physical
access to a device, the addresses could be obtainable by attaching hardware probes to the memory bus. In addition, a compromised operating system
or a malicious co-tenant in the cloud can obtain the addresses through cache-probing and other side-channel attacks~\cite{ristenpart09,razavi16,maurice17}.
Such leaks effectively break the confidentiality promise of enclaves~\cite{sgxsidechannels,malwareguard,sgxcacheattacks,liu2015last,sanctum,brasser17,moghimi17}.

In this paper, we explore Oblivious RAM (ORAM) designs that prevent these
information leaks under the constraints of modern SGX processors.
ORAMs completely eliminate information leakage through memory
address traces, but most ORAMs are a poor fit for enclave,
because they have high constant overhead factors~\cite{DBLP:journals/algorithmica/Paterson90, tree_based_orams,DBLP:conf/ndss/StefanovSS12} or require large private memories~\cite{pathoram},
which are not available in the enclave context.
We address these limitations with a new hierarchical ORAM, the Pyramid ORAM.
Our construction differs from previous hierarchical schemes
as it uses simpler building blocks that incur smaller constants in terms of complexity.
We instantiate each level of the hierarchy with our variant of multilevel adaptive hashing~\cite{hashtable},
that we call Zigzag hash table.
The structure of a Zigzag hash table allows us to avoid
the kind of expensive oblivious sorting that slowed down previous schemes~\cite{KLO,gm-paodor-11,GoldreichO96}.
We propose a new probabilistic routing network as a building block for constructing a Zigzag hash table
obliviously.
The routing network runs in $n\log n$ steps and does not bear ``hidden constants''.

Hierarchical ORAMs, despite being sometimes regarded inferior to tree-based schemes due to complex rebuild phases, have specific performance characteristics that make them the preferred choice for many applications.
Inherently, most accesses to a hierarchical ORAM are relatively cheap (i.e., when no or only few levels need to be rebuilt), while a few accesses are expensive (i.e., when many levels need to be rebuilt). In contrast, in tree-based schemes, all accesses cost the same. Since, in hierarchical ORAMs, the schedule for rebuilds is fixed, upcoming expensive accesses can be anticipated and the corresponding rebuilds can be done earlier, e.g., when the system is idle, or concurrently~\cite{ccsw,pis}. This way, high-latency accesses can be avoided.
To reflect this, we report \emph{online} and \emph{amortized} access costs. The former measures access time without considering the rebuild and the latter measures the online cost plus the amortized rebuild cost.
{Furthermore, hierarchical schemes can naturally support resizable
arrays and do not require additional techniques compared to tree-based schemes (as also noted in~\cite{dynamicORAM}).}

We present an efficient x64-optimized implementation of Pyramid ORAM that uses only CPU registers as private memory.
We compare our implementation with our own implementation of Circuit ORAM~\cite{circuit},
a state-of-the-art tree-based ORAM scheme that also uses constant memory,
on modern Intel Skylake processors.
Although the predecessor of Circuit ORAM, Path ORAM~\cite{pathoram}, has been successfully
used in the design of custom secure hardware
and memory controllers~\cite{ascend,phantom,Fletcher:2015:LLH:2860695.2860768,Ren:2013:DSE:2485922.2485971,Fletcher:2015:FON:2694344.2694353}, Circuit ORAM has better
asymptotical complexity when used with small private memory.

Pyramid ORAM has better online asymptotical complexity and outperforms Circuit ORAM in practice in many benchmarks, even without optimized scheduling of expensive shuffles. For typical configurations, Pyramid ORAM is at least {8x} faster than Circuit ORAM for {99\%} of accesses. 
Finally, although the best known hierarchical ORAM~\cite{KLO} has better asymptotical complexity,
Pyramid ORAM has significantly lower constant overhead factors, making it the preferred choice
in practice.

In summary, our key contributions are:
\begin{itemize}
\item A novel hierarchical ORAM, \ourscheme{}, that
asymptotically dominates tree-based schemes when comparing
online overhead of an oblivious access;
\item Efficient implementations of \ourscheme{} and Circuit ORAM for Intel SGX;
\item The first experimental results of a hierarchical ORAM in Intel SGX. Our results show that \ourscheme{} outperforms Circuit ORAM on small datasets and has significantly better latency for 99\% of accesses. 
\end{itemize}

\begin{table*}
\begin{center}
\small
\caption{Asymptotical performance of ORAM schemes for $\N$ elements
with private memory of constant size and elements of size $\D = \Omega(\log \N)$. All logarithms use the base of 2.}
\renewcommand{\arraystretch}{1.2}
\begin{tabular}
{ l | c | c | l}
\multirow{ 2}{*}{Scheme} &  \multicolumn{2}{c|}{Bandwidth cost} & Server Space \\ 
\cline{2-3}
&  Online & Amortized  & Overhead \\
\hline
Binary Tree ORAM~\cite{tree_based_orams} (Trivial bucket) & \multicolumn{2}{c|}{$\D(\log \N)^2+(\log \N)^4$} & $\log \N$\\
Path ORAM (SCORAM)~\cite{scoram} &  \multicolumn{2}{c|}{$\D\log\N + (\log\N)^3 \log \log \N$}  & 1 \\
Circuit ORAM~\cite{circuit} & \multicolumn{2}{c|}{$\D\log\N + (\log\N)^3$} & 1 \\
\hline
Hierarchical, GO~\cite{GoldreichO96} &  $\D(\log \N)^2$ & $\D(\log \N)^3$ & $\log \N$ \\
Hierarchical, GM~\cite{gm-paodor-11} & $\D(\log \N)$ & $\D(\log \N)^2$ & 1 \\
\cline{2-3}
Hierarchical, Kushilevitz~\textit{et al.}~\cite{KLO} &  \multicolumn{2}{c|}{$\D (\log \N)^2 / \log \log \N$} & 1 \\
\hline
\hline
\textbf{\ourscheme{}}  &$\D (\log \N)  (\log \log \N)^2$   & $\D (\log \N)^2  (\log \log \N)^2$  & $(\log \log \N)^2$  \\
\end{tabular}
\label{tbl:cmp}
\end{center}
\end{table*}

\section{Preliminaries}

\paragraph{Attacker Model}
Our model considers a trusted processor and a trusted client program running inside
an enclave. We assume the trusted processor's implementation is correct
and the attacker is unable to physically open the processor package to extract secrets.
The processor has a small amount of private trusted memory: the processor's registers;
we assume the attacker is unable to monitor accesses to the registers.
The processor uses an untrusted external memory to store data; we also
refer to the provider of this external memory as the untrusted server as this
terminology is often used in ORAM literature; in our implementation
this is simply the untrusted RAM chips that store data. We assume the attacker
can monitor all accesses of the client to the data stored in the external memory.
We assume the processor encrypts the data before sending it to external memory
(this is implemented in SGX processors~\cite{sgx1}), but the addresses where the
data is read from or written to are available in plaintext.
We assume the processor's caches are shared with programs under the control
of the attacker, for instance malicious co-tenants running on the same CPU~\cite{ristenpart09,malwareguard}
or a compromised operating system~\cite{sgxsidechannels}.
Hence, we assume that an attacker can observe client memory accesses at cache-line granularity.
This model covers scenarios where SGX processors run in an untrusted cloud environment where they are shared amongst many tenants, the cloud operating system may be compromised, or the
cloud administrator may be malicious.
All side channels that are not based on memory address traces, including power analysis, accesses to disks
and to the network, as well as channels based on shared microarchitectural elements
inside the processor other than caches, (e.g., the branch predictor~\cite{branch_shadowing}) are outside the scope of this paper.

\paragraph{Definitions}
An element $e$ is an index-value pair
where $e.\key$ and $e.\val$ refer to its key and value, respectively.
A \emph{dummy} element is an element of the same size as a real element.
Elements, real or dummy, stored outside of the client's private memory
are encrypted using semantically secure encryption
and are re-encrypted whenever the client reads them to
its private memory and writes back.
Hence, encrypted dummy elements are indistinguishable from real elements.
Our construction relies on hash functions that are instantiated using pseudo-random functions
and modeled as truly random functions.
We say that an event happens with negligible probability
if it happens with probability at most $\frac{1}{2^\secsize}$, for
a sufficiently large $\secsize$, e.g., 80.
An event happens with very high probability (w.v.h.p.)
if it happens with probability at least $1 - \frac{1}{2^\secsize}$.

\paragraph{Data-oblivious Algorithm} 
To protect against the above attacker,
the client can run a \emph{data-oblivious}
counterpart of its original code.
Informally, a data-oblivious algorithm
appears to access memory independent of the content of
its input and output.

Let $\mathcal{A}$ be an algorithm that
takes some input and produces an output.
Let $\rho$ be the public information
about the input and the output (e.g., the size of the input).
Given the above attacker model,
we assume that all memory addresses
accessed
by $\mathcal{A}$ are observable.
We say that $\mathcal{A}$
is data-oblivious, i.e., secure in the outlined setting,
if there is a simulator $\Sim{\mathcal{A}}$
that takes \emph{only}~$\rho$
and produces a trace of memory accesses that
appears to be indistinguishable
from $\mathcal{A}$'s.
That is, a trace of a data-oblivious algorithm
can be reconstructed using $\rho$.

\paragraph{Oblivious RAM}
ORAM is a data-oblivious RAM~\cite{GoldreichO96}.
We will distinguish between \textit{virtual} and \textit{real accesses}.
A virtual access is an access the client intends to perform on RAM,
that is, an access it wishes to make oblivious.
The real access is an access that is made to the data structures
implementing ORAM
in physical memory, i.e., those that appear on the trace produced by an algorithm.

\paragraph{Performance}
We assume that the client uses external memory (i.e., memory at the server) to store elements of size $\D$ (e.g., 64 bytes)
and it can fit a constant number of such elements in its private memory.
{Similar to other work on oblivious RAM~\cite{pathoram,circuit}, we
assume that~$\D$ is of size~$\Omega(\log \N)$, that is, it is sufficient to store a tag of the element.
During each access the client can read a constant number of elements into its private memory,
decrypt and compute on them, re-encrypt, and write them back.
When analyzing the performance of executing an algorithm that operates on external memory,
we are interested in counting the number of bits of the external memory the client must access.
Similar to previous work on ORAM~\cite{circuit,pathoram}, we call this measurement \emph{bandwidth cost}.

For ORAM schemes,  we distinguish between \emph{online}
and \emph{amortized bandwidth cost}.
The former measures the number of bits to obtain a requested element, while the latter measures
the amortized bandwidth for every element.
This distinction is important especially for hierarchical schemes that periodically stall to perform
a reshuffle of the external memory before they can serve the next request.
However, such reshuffles happen after a certain number of requests
are served and their cost can be amortized
over preceding requests that caused the reshuffle.

In Table~\ref{tbl:cmp} we compare the overhead of \ourscheme{} with other ORAM schemes that use constant private memory.
Though \ourscheme{} uses more space at the server than tree-based schemes,
it has better asymptotical online performance for blocks of size $O((\log \N/\ \log\log\N)^2)$.
\ourscheme{} is less efficient than hierarchical schemes by
Kushilevitz~\textit{et al.}~\cite{KLO} and Goodrich and Mitzenmacher~\cite{gm-paodor-11}.
However, as we explain in Section~\ref{sec:relatedwork},
our scheme relies on primitives that are more efficient in practice than those used in~\cite{KLO,gm-paodor-11}.
For example, we do not rely on the AKS sorting network.

\section{Hierarchical ORAM (HORAM)}
\label{sec:hierarchical}

Our \ourscheme{} scheme follows the same general structure as other hierarchical ORAM (HORAM)
constructions~\cite{GoldreichO96,wsc-bcomp-08,KLO,gm-paodor-11,GMOT12}.
We outline this structure below and describe how we instantiate it
using our own primitives, \nameref{sec:network} and \nameref{sec:ZHT}, in the next sections.

\paragraph{Data Layout}
HORAM abstractly arranges the memory at the server
into $O(\log \N)$ levels, where $\N$ is the number of elements
one wishes to store remotely and access obliviously later.
Though~$\N$ does
not have to be specified a priori and the structure can grow to accommodate more elements,
for simplicity, we assume~$\N$ to be a fixed power of 2.

The first two levels, $\Level_0$ and $\Level_1$, can store up to $\firstLevelSize$
elements each, where $\firstLevelSize\ge 2$.
The third level $\Level_2$, can store up to $2\firstLevelSize$ elements
and
the $i$th level stores up to $2^{i-1} \firstLevelSize$.\footnote{The construction by Kushilevitz~\textit{et al.}~\cite{KLO} uses
a slightly different layout from the one described here. In particular,
it has less levels and lower levels can hold $\log \N \times$ the size of the previous level instead of~$2\times$.}
The last level, $l$, can fit all~$\N$ elements (i.e.,
$l = \log \N/\firstLevelSize$).
Note that we only mentioned the number of real
elements that each level can fit.
In many HORAM constructions,
each level also contains dummy elements
that are used for security purposes. Hence,
the size of each level is of some fixed size
and is often larger than the number of real elements the level can fit.
The data within each level is stored encrypted (implicitly done by the hardware in the case of SGX), hence,
real and dummy elements are indistinguishable.

The first level $\Level_0$ is implemented as a list.
The search over $\Level_0$ is a simple scan.
Each consequent level, $\Level_{i\ge1}$, is implemented
as a hash table to ensure fast search time (e.g., a cuckoo hash table~\cite{KLO,gm-paodor-11,GMOT12,Pinkas2010} or a hash table
with buckets~\cite{GoldreichO96}).

\paragraph{Setup}
During the setup, the client inserts $\N$ elements into the last level $\Level_l$
and uploads the first $l-1$ empty levels and $\Level_l$ to the server.
As mentioned before, this is not required and the client can populate the data structure
by writing elements obliviously one at a time.

\paragraph{Access}
During a virtual access for an element with $e.\key$, the client performs
the following search procedure across all levels.
If level $i$ is empty, it is skipped.
If the element was not found in levels $0,\ldots,i-1$, then
the client searches for $e.\key$ in $\Level_i$.
If $\Level_i$ has the element, it is removed from $\Level_i$,
cached, and the client continues accessing levels $i+1, i+2, \dots$
but for dummy elements.
If~$\Level_i$ does not have the element, the client
searches for $e.\key$ in level $i+1$ and so on.
After all levels are accessed, the cached $e$ is appended to $\Level_0$
(even if $e$ was found in some earlier location in $\Level_0$ from which
it was removed).

\paragraph{Rebuild}
After $\firstLevelSize$ accesses,
$\Level_0$ contains up to $\firstLevelSize$ elements, where some may be dummy
if the same element was accessed more than once. This triggers the rebuild phase:
$\Level_0$ is emptied by inserting its \emph{real} elements into~$\Level_1$.
After the next $\firstLevelSize$ accesses, $\Level_0$ becomes full again.
Since $\Level_1$ is already full, real elements from both~$\Level_0$ and~$\Level_1$ are
inserted into~$\Level_2$.
In general, after every $2^i \firstLevelSize < \N$ accesses levels
$0,1, \ldots, i$ are emptied into $\Level_{i+1}$.
After $\N$ accesses levels $0,1,\ldots,l-1$ are all full.
In this case, a new table, $\Level'_{l}$, of the same size as $\Level_l$, is created.
Elements from~$\Level_0, \Level_1, \ldots, \Level_{l}$ are inserted
into~$\Level'_{l}$ and emptied, and $\Level_{l}$ is replaced with $\Level'_{l}$.
The same process repeats for the consequent accesses.
Since the rebuild procedure is a deterministic function of the number of accesses performed so far,
it is easy to determine which levels need rebuilding for any given access.

The insertion of real elements
from levels $0,1, \ldots, i$ to $i+1$ has to be done carefully
in order to hide the number of real and dummy elements
at each level. For example, knowing that there was only one real element
reveals that all past accesses were for the same element.
The data-oblivious insertion (or rebuild) process depends on the specific HORAM construction
and the data structure used to instantiate each level.
Many constructions
rely on an oblivious sorting network~\cite{batcher,AKS} for this.

\paragraph{Security}
A very high-level intuition behind the security of an HORAM
relies on the following: (1)
level $i$ is searched for every element at most once
during the accesses between two sequential rebuilds of level~$i$;
(2) searches for the same element before and after a rebuild
are uncorrelated;
(3) the search at each level
does not reveal which element is being searched for.

\paragraph{\ourscheme{} Overview}
HORAM's performance and security rely on two critical
parts: the data structure used to instantiate levels $\Level_1, \ldots, \Level_l$
and the rebuild procedure.
In this paper we propose \ourscheme{}, an HORAM instantiated with a novel data structure and
a rebuild phase tailored to it.
In Section~\ref{sec:ZHT} we devise a new hash table, a Zigzag Hash Table, and use it to instantiate each level
of an HORAM.
To avoid using an expensive oblivious sorting network when rebuilding
a Zigzag Hash Table, in Section~\ref{sec:ozht} we design a setup phase using
a Probabilistic Routing Network (Section~\ref{sec:network}).
The combination of the two primitives and their small constants result in an HORAM that is more efficient in practice
than previous HORAM schemes.

\section{Probabilistic Routing Network (PRN)}
\label{sec:network}
\begin{figure}[t!]
\centering
    \includegraphics[width=\columnwidth]{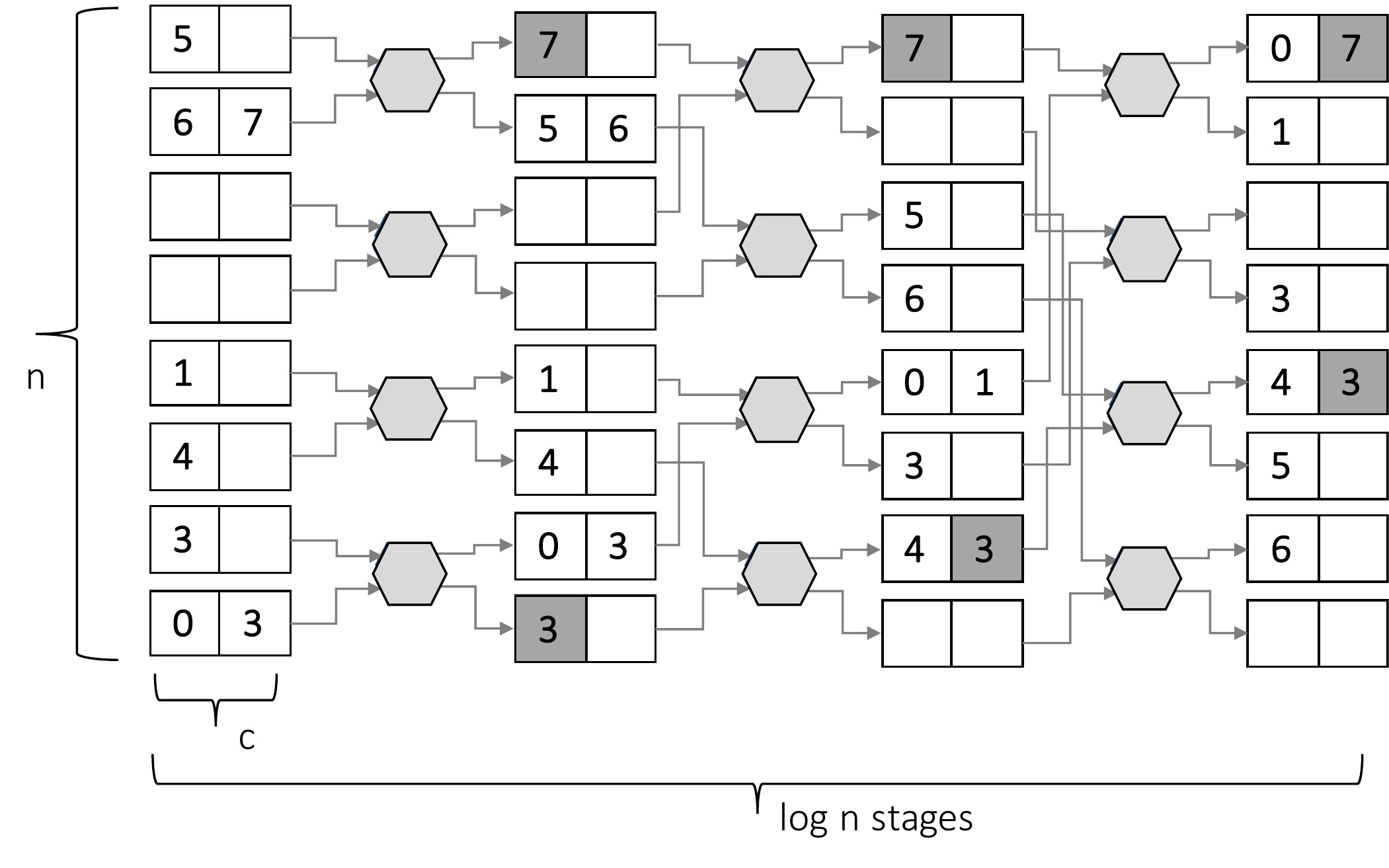}
\caption{Routing network for $\n=8$ elements where each element
is indicated in the input table by its destination position (e.g., the destination
of the element in the first bucket is 5 and destinations of the two elements in
the second bucket are 6 and 7).
The elements that could not be routed in the second stage (i.e., tagged $\false$ by the network) are greyed out.
}
\label{fig:network}
\end{figure}

In this section, we describe a Probabilistic Routing Network
that we use in the rebuild phase
of our HORAM instantiation.
However, the network as described here,
with its data-oblivious guarantees and efficiency,
may be of an independent interest.

A PRN takes a table of $\n$ buckets with $c$ slots in each bucket.
The bucket size is sufficiently small compared to $n$.
For example, in \ourscheme{} $c$ is set to $O(\log \log \n)$.
The input table has at most $\n$ real
elements labeled with indices $0,1, \ldots, \n-1$ denoting their destination buckets.
The labels do not have to be stored alongside the elements and can be determined
using a hash function, for example.
There can be more than one element in the input bucket
and, since a hash function is not a permutation function,
two or more elements may be assigned to the same destination bucket.

A PRN's goal is to route a large subset of elements to the correct output buckets.
In the output, elements in that subset are associated with the tag $\true$, others with $\false$.
Given the destination labels, the network's output can be used for efficient element lookups.
Such a lookup is \emph{successful}, if an element was routed to its correct destination (i.e., has tag $\true$).

For the security purposes of HORAM, we wish to design a network
that is data-oblivious in the following sense.
The adversary that observes the routing
as well as element accesses to the output
learns nothing about the placement of the elements
in the output as long as it does not learn whether
accesses are successful or not.

\paragraph{Algorithm}
We now describe a routing network with the above properties.
The routing network consists of $\log \n$ stages where the input at each stage, except the first stage, is the output of the previous stage.
The output of the last stage is the output of the network.
Besides the destination bucket indices, every element also carries a Boolean tag that is set to $\true$ for every real element in the original input.
(However, depending on the use case, input elements may be tagged $\false$, for example, if they are dummy elements.)
Once the tag is set to $\false$, it does not change throughout the network.
Crucially, only elements that are in their correct destination buckets in the output are tagged $\true$.

At each stage $i$, for $1 \leq i \leq \log n$, we perform a re-partition operation between each pair of buckets with indices that differ only by the bit in position $i$ (starting from the lowest bit).
(Note that here indices correspond to indices of the input buckets, not the indices of the destinations of the elements inside.) For example, for $n=8$,
when $i=1$ we re-partition the following pairs of buckets: $(000, 001), (010, 011), (100, 101) \ \mathrm{ and } \ 
 \ (110, 111)$
 and when $i = 3$, buckets: $(000, 100), (001,101), (010,110)$ and $(011,111)$. 
The re-partition operation itself moves items between the two buckets (number them one and two) and works as follows:
For elements with tag $\true$, we consider the $i$th bit of their bucket index.
If this bit is set to 0, the element goes to the first bucket,
and to the second bucket, otherwise. 
As there may be more than $c$ elements assigned the same bit at position $i$,
the network chooses $c$ of them at random, and retags the extra/spilled elements with $\false$.
All elements tagged $\false$, including those that arrived with the $\false$ tag from the previous stage,
are then allowed to go in either bucket (note that there is always space for all elements). 
The effect of the network is that at stage $i$, every element still tagged $\true$ will have the first $i$ bits match that of the destination bucket.
Therefore, after $\log n$ stages elements will either be tagged $\false$, or be in the correct bucket.
(An illustration of the routing network for $\n=8$ is given in Figure~\ref{fig:network}.)

\paragraph{Bucket Partitioning}
We now describe how to perform the re-partition of two buckets.
If the client memory is $2\BSize$ then the re-partition is trivial.
Otherwise, we use Batcher's sorting network~\cite{batcher} as follows.
The two buckets are treated as one contiguous range of length $2c$
that needs to be sorted ordering the elements as follows: (1) elements with bit~0 at the~$i$th position,
(2) empty elements and elements tagged~$\false$, and (3) elements with bit~1 at the~$i$th position.
Sorting will therefore cause a partition, with empty and $\false$-tagged elements in the middle, and other elements in their correct buckets.

\paragraph{Analysis}
In this section we summarize the performance and security properties of the network.

\newcommand{\Sort}{\mathsf{Sort}}

\begin{theorem}
The probabilistic routing network
operates in $(\n / 2) \log \n \times \Sort(2c)$ steps, where~$\Sort(c)$ is time to sort $c$ elements obliviously.
\label{thm:singularRN1}
\end{theorem}
The theorem follows trivially: the network has $\log \n$ stages where each stage consists of repartitioning
of~$\n/2$ pairs of buckets.

\begin{corollary}
The probabilistic routing network
operates in $O(c (\log c)^2 \n \log \n)$ steps
using constant private memory and Batcher's sorting network for repartitioning bucket pairs.
\label{cor:singularRN2}
\end{corollary}

The PRN is more efficient than Batcher's sorting network, that runs in $O(\n(\log\n)^2)$ steps,
at the expense of not being able to route all the elements.
Though the network cannot produce all $\n!$ permutations,
it gives us important security properties.
As long as
the adversary does not learn which output elements
are tagged $\true$ and $\false$ it does not learn the position
of the elements in the output. 
Hence, when using the network in the ORAM setting
we will need to ensure that the adversary does not learn whether an element
reached its destination or not.
We summarize this property below.

\begin{theorem}
The probabilistic routing network and lookup accesses to its output are
data-oblivious as long as the adversary does
not learn whether lookups were successful or not.
\label{thm:onetwork}
\end{theorem}
\begin{proof}
The accesses made by the routing network to the input are independent of the input
as they are deterministic and the re-partition operation is either done in private memory or using a data-oblivious sorting
primitive. Hence, we can build the simulator $\Sim{\route}$ that,
given a table of dummy elements,
performs accesses according to the network.
The adversary cannot distinguish $\Sim{\route}$
from the real routing even if it knows the location of the elements
in the input.

During the search, the adversary is allowed to request lookups for the elements
from the original input and observe the accessed buckets.
For data-obliviousness, the lookup simulator is notified about the lookup
but not the element
being searched. The simulator then accesses a random bucket
of the output generated by~$\Sim{\route}$.

The accesses of the simulator and the real algorithm diverge only
in the lookup procedure. The accesses generated by the simulator
are from a random function, while real lookups access elements
according to a hash function used to route the elements.
Since the adversary is not notified
whether a lookup was successful or not, it cannot distinguish
whether simulator's function could have been used to
successfully allocate the same elements as the real routing network.
Hence, the distribution of lookup accesses by the simulator is indistinguishable
from the real routing network.
\end{proof}

We defer the analysis of the number of elements
that are routed to their output bucket to the next section
as it depends on how elements are distributed
across input buckets.
For example, we assume random allocation
of the elements in the input in~\ourscheme{}.

\begin{table}[b]
\caption{Notation.}
\begin{center}
\begin{tabular}
{|c || l|}
\hline
$\N$ & Number of elements in ORAM\\
\hline
$\n$ & Capacity of a \ourhashtable{} (\ourhashtableS{}) \\
\hline
$\NTable$ & Number of tables in a \ourhashtableS{}\\
\hline
$\BSize$ & Bucket size of tables $H_1, \ldots, H_k$ in a \ourhashtableS{}\\
\hline
\end{tabular}
\end{center}
\end{table}

\section{Zigzag Hash Table (ZHT)}
\label{sec:ZHT}

In this section we describe a data structure, called Zigzag Hash Table (ZHT),
that we use to instantiate the
levels of our hierarchical ORAM.
A ZHT is a variation of multilevel adaptive hashing~\cite{hashtable}
but instantiated with tables of fixed size and buckets of size $c$ (see Section~\ref{sec:relatedwork} for a detailed comparison).
We begin with the non-oblivious version of a ZHT
and then explain how to perform search and
rebuild obliviously.

\begin{figure}[t!]
\centering
    \includegraphics[scale=0.45]{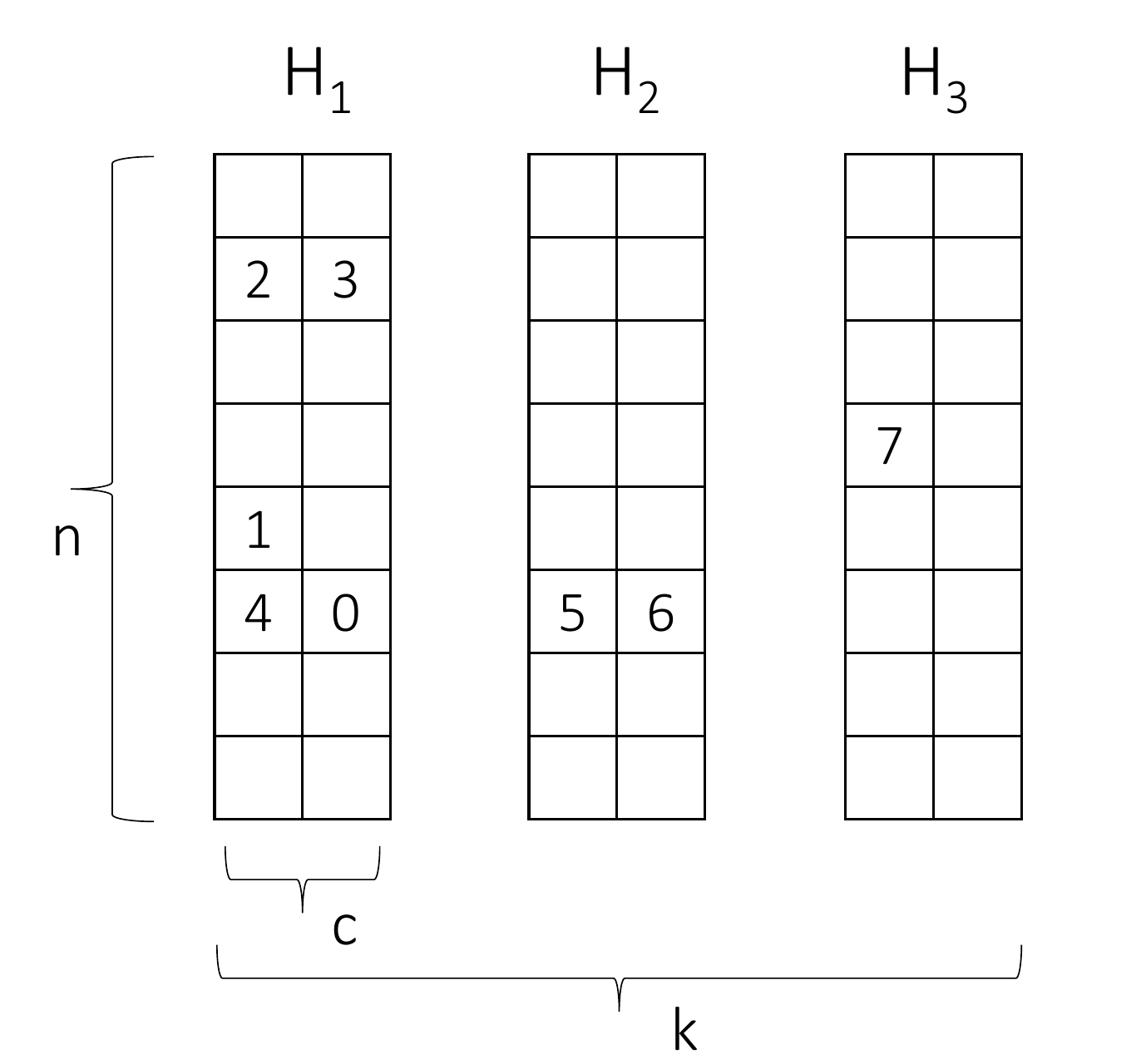}
\caption{A Zigzag hash table for $n=8$ elements with $k=3$ tables
where each table has $8$ buckets of size $c=2$.
Elements 0--7 were inserted in order,
with hash functions defined as: $h_1(1) = 4$, $h_1(2) = h_1(3) = h_1(5) = h_1(6) = h_1(7) = 1$, $h_1(4) = h_1(0) = 5$,
$h_2(5) = h_2(6) = h_2(7) = 5$ and $h_3(7) = 3$.
Zigzag path of element 7 is defined as $\p_7 = \{1,5,3\}$.}
\label{fig:ZHT}
\end{figure}
A ZHT is a probabilistic hash table that can store up to $\n$ key-value elements
with very high probability.
We refer to $\n$ as the capacity of a ZHT, that is, the largest number of elements
it can store.
ZHT consists of~$k$ tables, $H_1, H_2, \ldots, H_k$.
Each table contains $\n$ buckets, with each bucket
containing at most $\BSize$ elements.
Each table is associated with a distinct hash function $h_i$ that maps element keys to buckets.
The hash functions are fixed during the setup, before the elements arrive.

An element~$e$ is associated with an (abstract) 
\emph{zigzag path} defined by the series of bucket locations in each table: $h_1(e.\key)$, $h_2(e.\key),$ $\ldots, h_k(e.\key)$.
Elements are inserted into the table by accessing the buckets along the path, placing the element in the first empty slot.
If the path is full, the insertion procedure has failed.
When referring to the insertion of several elements, we use
$\throw(A, H_1$, $H_2$, $\ldots, H_k)$ that
takes elements in the array $A$ and inserts them according to their
zigzag paths.
Once inserted, elements are found by scanning the path and comparing indices of elements in the buckets until the element is found.
(See Figure~\ref{fig:ZHT} for an illustration of a ZHT.)

\begin{definition}[Valid zigzag paths]
Let $\p_1, \p_2, \ldots, \p_n$ be $n$ zigzag paths.
If every element $e_j$ can be inserted into a ZHT according to locations in the path $\p_j$,
then the paths are valid.
\end{definition}

If $n$ elements are successfully inserted in a ZHT,
we say that the zigzag paths associated with the elements
are valid.
We will use ZHT$_{n,k,c}$ to denote a ZHT with parameters $k$ and $c$ sufficient
to avoid overflow with high probability when inserting $n$ elements.
We summarize the parameters and the performance of a ZHT below.

\begin{theorem}[ZHT Performance]
Insertion and search of an element in a ZHT$_{n,k,c}$ takes $O(k\BSize)$ accesses.
Insertion of $n$ elements fails with negligible probability when
$k = O(\log \log n)$ and $c = O(\log \log n)$.
\label{thm:ZHT}
\end{theorem}
\begin{proof} \emph{(Sketch.)}
The analysis follows Theorem 5.2 from~\cite{Adler:1995:PRL:225058.225131}, adopted for our scenario.
We note that the analysis in~\cite{Adler:1995:PRL:225058.225131} is for re-throwing elements into the same table
and not into an empty table, as is the case in our data structure.

We refer to an element of a table as \emph{overflowed}
if the bucket it was allocated to already had $\BSize$ elements
in it.

For each table $H_i$, we are interested in measuring $\inc_i$, the number of incoming elements, and $\out_i$, the number of outgoing elements.
The former comprises the elements that did not fit
in tables $1,2,\ldots,i-1$.
The number of outgoing elements comprises elements that did not fit into table $i$
when throwing $\inc_i$.
Except for the first phase, the number
of incoming elements into $H_i$ is determined by the number of outgoing elements
of $H_{i-1}$.
Hence, to determine $k$
we need to bound the number of outgoing elements
at each level.

We first show that one round of throw
is required
to bound the number of overflowed elements to be at most $n/(2e)$
(Lemma~\ref{lemma:firstounds}).
We then show that as long as the number of elements
arriving at the table is less than $n/(2e)$
only $O(\log \log n)$ tables are required to fit all $n$ elements
(Lemma~\ref{lemma:lastrounds}).
\end{proof}

We note that buckets that hold at most one element (i.e., $c =1$)
are sufficient to bound the number of tables, $k$, to be $\log \log n + O(1)$
~\cite{Adler:1995:PRL:225058.225131,hashtable}.
However, as will be explained in the following sections,
in order to insert elements into a ZHT
obliviously we require $c$ to be $O(\log \log n)$.

As described so far, the insertion procedure of a ZHT does not provide security properties
one desires for ORAM (e.g., observing insertion of $n$ elements and a search for element $e$ reveals when $e$
was inserted). 
In the next section, we make small alternations to the
insertion and search procedures of a ZHT and summarize the security properties a ZHT already has.
In particular, we will show that one can obliviously search for elements in a ZHT, as long
as it was constructed privately.
Then, in Section~\ref{sec:ozht} we show how to insert~$n$ elements into a ZHT obliviously
using small private~memory.

\subsection{Data-oblivious properties of ZHT}
We make the following minor changes to a ZHT
and its insertion and search procedures.
We fill all empty slots of a ZHT with dummy elements.
Recall that in the ORAM setting the data structure is encrypted and dummy elements
are indistinguishable from real.
Then, when inserting an element, instead of stopping once the first empty slot is found,
we access the remaining buckets on the path,
pretending to write to all of them.
Similarly, during the search, we do not stop when the element is found
and access all buckets on its zigzag path.
Hence, anyone observing accesses of an insert (or search)
observe the zigzag path but do not know where on the path the element
was inserted (or found). 
We define dummy counterparts of insertion and search procedures that simply
access a random bucket in each table of a ZHT, performing fake writes and reads.

Given the above changes, the search for $n$ elements
is data-oblivious as long as every element is searched for only once
and assuming that an adversary does not observe how the table was built.

\begin{lemma}[ZHT Data-Oblivious Search]
Searching for~$n$ distinct elements in a ZHT$_{n,k,c}$
is data-oblivious.
\label{lemma:ZHTOblivious1}
\end{lemma}
\begin{proof}
We can build a simulator that given $n,k$ and $c$
produces a sequence of accesses to its tables indistinguishable
from locations accessed by a ZHT algorithm.

Consider an algorithm (with large private memory) that takes~$n$ elements
and inserts them one by one into an empty ZHT.
If it cannot insert all $n$ elements, it fails.
Otherwise, during a search request for an element key, $e.\key$,
it accesses the locations that correspond to the zigzag path
defined by $e.\key$ and the hash functions of the~ZHT.

The simulator,  $\Sim{\search}$, sets up~$k$ tables with $n$ buckets of size~$c$ each.
For every search, the simulator performs a dummy search.
We say that a simulator fails if after simulating searches for $n$ elements, it
has produced accesses to ZHT tables that do not form $n$ valid zigzag paths.

In each case the adversary chooses $n$ original elements. It then requests a search
for one of the $n$ elements that it has not requested before.
A ZHT algorithm is given the key of the element, while the simulator is invoked to perform a dummy search.
In both cases, the adversary is given the locations accessed during the search.
If neither the simulator nor the ZHT algorithm fail,
the accessed locations are indistinguishable as long as ZHT's hash functions are indistinguishable from truly random functions.
The failure cases of the algorithms allow the adversary to distinguish the two as failures happen
at different stages of the indistinguishability game.
However, both failures happen with negligible probability due to the success rate of a ZHT$_{n,k,c}$.
\end{proof}

\begin{lemma}
Searching for $n$ distinct elements in a ZHT$_{n,k,c}$
is data-oblivious
even if the adversary knows the order of insertions.
\label{lemma:ZHTOblivious2}
\end{lemma}
\begin{proof}
The proof follows the proof of Lemma~\ref{lemma:ZHTOblivious1},
except now the adversary controls the order in which the ZHT algorithm
inserts the elements. As before, we use $\Sim{\search}$ that is not given the actual elements and the adversary does not observe accesses of the insertion phase.

Let $\p_i$ be the zigzag path of element $e_i$. Notice that
$k$ hash functions that determine zigzag paths are chosen
before the elements are inserted into an empty ZHT.
Hence, the order in which~$e_i$ is inserted in the ZHT influences where on the path~$e_i$ will appear (e.g., if~$e_i$ is inserted first, then it will be placed into~$H_1$, on the other hand, if it is inserted last it may not
find a spot in $H_1$ and could even appear in $H_k$) but not the path it is on.
Since the adversary does not learn where on the path an element was found, the rest of the proof is
similar to the proof of Lemma~\ref{lemma:ZHTOblivious1}
\end{proof}

We capture another property of a ZHT
that will be useful when used inside of an ORAM:
locations accessed during the insertion
appear to be independent of the elements being thrown.
Moreover, as we show in the next lemma,
the adversary cannot distinguish whether
thrown elements were real or dummy.
\begin{lemma}
The algorithm $\throw(A, H_1, H_2, \ldots, H_k)$
is data-oblivious where $H_1, H_2, \ldots, H_k$
are hash tables of ZHT$_{n,k,c}$
and the size of $A$ is at most $\n$.
\label{lemma:ZHTObliviousInsert}
\end{lemma}
\begin{proof}
The adversary chooses up to $n$ distinct elements and
requests $\throw$.
The algorithm fails if it cannot insert all elements
in $k$ tables.
The simulator $\Sim{\throw}$ is given $A$ that contains only dummies.
It accesses a random bucket in each of the $k$ tables for every
element in~$A$.

The adversary can distinguish the two algorithms only when
$\throw$ fails or when $\Sim{\throw}$ generates invalid zigzag paths for~$A$.
According to Theorem~\ref{thm:ZHT}, this happens with negligible probability.
\end{proof}

\begin{lemma}
Consider $\throw(A, H_1, H_2, \ldots, H_k)$ where $A$ contains $n_\real \le n$ real elements and $n_\dummy$ dummy elements,
where a real element is inserted as usual while on a dummy element each table is accessed at a random index.
The algorithm $\throw$ remains data-oblivious.
\label{lemma:throw_with_dummy}
\end{lemma}
\begin{proof}
We define an adversary that chooses a permutation of $n_\real$ real elements
and $n_\dummy$ dummy elements, that is, it knows the locations of the
real elements.
We build a simulator that is given an empty ZHT$_{n,k,c}$ table. The simulator
accesses a random bucket in each table of the ZHT on every element in~$A$.

The adversary can distinguish between the simulator and the $\throw$ algorithm if (1) the algorithm fails, which happens when the algorithm cannot insert all real elements successfully
or (2) the simulator generates invalid zigzag paths when performing
a fake insertion for the input locations that correspond to real elements
(recall that the simulator does not know the locations of real elements).
Both events correspond to the failure of the ZHT construction, which happens with negligible probability (Theorem~\ref{thm:ZHT}).
\end{proof}

\subsection{Data-oblivious ZHT Setup}
\label{sec:ozht}

A Zigzag hash table allows one to search elements obliviously if insertion
of these elements was done privately.
Since the private memory of the client is much smaller than $\n$ in our scenario,
we develop
an algorithm that can build a ZHT relying on small memory and accessing the server memory
in an oblivious manner.
That is, the server does not learn the placement of the elements
in the final~ZHT.

The algorithm starts by placing elements in the tables $H_1, H_2, \ldots, H_k$
using $\throw$ procedure from~Section~\ref{sec:ZHT} according to
random hash functions.
The algorithm then proceeds in $k$ phases.
In the first phase, $H_1$ is given as an input to the probabilistic routing
network where the destination bucket of every element in $H_1$
is determined by $h_1$.
The result is a hash table
with some fraction of the elements in the right buckets and some overflown elements (tagged as $\false$).
The algorithm then pretends to throw \emph{all} elements of~$H_1$ into
tables $H_2, H_3, \ldots, H_k$ using yet another set of random functions.
However, only overflown elements, that is elements that were not routed to their correct bucket in $H_1$,
are actually inserted into tables $H_2, H_3, \ldots, H_k$.
For empty entries and elements that are labeled $\true$ in $H_1$,
the algorithm accesses random locations in $H_2, H_3, \ldots, H_k$
and only pretends to write to them.
In the end of this step every element that did not fit into~$H_1$
during the routing network appears in one of $H_2, H_3, \ldots, H_k$
at a random location.
In the second phase, $H_2$ goes through the network
where destination buckets are determined using $h_2$, and so on.
In general, every~$H_i$ goes through the network using~$h_i$,
then all elements of~$H_i$ are thrown into  tables $H_{i+1}, H_{i+2}, \ldots, H_k$,
where only elements that were not routed correctly are inserted.
(See Figure~\ref{fig:fullZHT} in Appendix for an illustration of the oblivious setup of a ZHT.)

The algorithm fails if an element cannot find a spot in any of the~$k$ tables.
This happens either during one of the throws of the elements
or during the final phase where not all elements of $H_k$ can be routed 
according
to~$h_k$.

\paragraph{Performance}

We now analyze the parameters and the performance of the ZHT setup
and then describe the security properties.

\begin{theorem}
Insertion of $n$ elements into a ZHT$_{n,k,c}$ as described in Section~\ref{sec:ozht} succeeds with
very high probability if $k = O(\log \log n)$ and $c = O(\log \log n)$.
\label{thm:ozht}
\end{theorem}
\begin{proof}
\emph{(Sketch.)}
The analysis is similar to the analysis in Theorem~\ref{thm:ZHT}
with the difference that the number of elements going into
level $H_{i+1}$ is higher due to additional overflows that happen during routing.
In Lemma~\ref{lemma:network_indep_stage},
we show that if $n$ is the number of input elements in hash table $H_i$, then the number of overflows from
routing is
bounded by $\log n \times$ the number of overflows of a $\throw$ into a single hash table.
Then, 
following similar proof arguments as Theorem~\ref{thm:ZHT},
we use Lemmas~\ref{lemma:firstounds2} and~\ref{lemma:lastrounds2} to bound the spill
at each phase and show that
if $k = O(\log \log n)$, then the insertion of~$n$ elements succeeds w.v.h.p.
\end{proof}

\begin{theorem}[Oblivious ZHT Setup Performance]
If oblivious ZHT insertion of $n$ elements does not fail,
it takes $O(nkc ((\log c)^2\log \n + kc))$ steps.
\label{thm:ozht_performance}
\end{theorem}
\begin{proof}
The initial throw takes $nkc$
and $n(k-1)c^2$, $n(k-2)c^2$, $\ldots$, $n c^2$ steps for the subsequent throws.
The complexity of applying a routing network to a single table is $O(c (\log c)^2 \n \log \n)$.
Since the algorithm performs $k$ routing networks,
the total cost is $O(k c (\log c)^2 \n \log \n + n(kc)^2) = O(nkc ((\log c)^2\log \n + kc))$.
\end{proof}
When used in \ourscheme{}, a ZHT is constructed from
all elements in previous levels of the hierarchy.
The levels include real elements as well as many dummy elements.
Hence, the initial throw of the setup algorithm may contain additional
dummy elements.
Throwing dummy elements does not change the
failure probability of the setup.
However, the initial throw becomes more expensive than
that in Theorem~\ref{thm:ozht_performance}
as summarized below.
\begin{corollary}[Oblivious ZHT Setup Performance]
Oblivious ZHT insertion of $n$ real elements and $(\alpha-1) n$ dummy elements
takes $O(nkc ((\log c)^2\log \n + \alpha kc))$, for $\alpha \ge 1$.
\label{corollary:dummy_real_zht}
\end{corollary}

\paragraph{Security}

We show that an adversary that observes the ZHT setup does
not learn the final placement of the elements across the tables.

\begin{theorem}
The insertion procedure of ZHT$_{n,k,c}$ described in Section~\ref{sec:ozht}
is data-oblivious.
\label{thm:oZHT}
\end{theorem}
\begin{proof}
An adversary is allowed to interact with the algorithm and a simulator
as follows:
choose up to $n$ elements and request the oblivious setup
of ZHT$_{n,k,c}$, then request a
search of at most $n$ elements (real or dummy) on ZHT$_{n,k,c}$.
It then repeats the experiment either for the same set of $n$ elements or a different one.

As before, the real algorithm fails if it cannot build ZHT$_{n,k,c}$.
The simulator $\Sim{\build}$ is given an array $A$ of $n$ dummy elements and
parameters $k$ and $c$.
It sets up ZHT$_{n,k,c}$
and runs $\Sim{\throw}(A, H_1$, $H_{2}, \ldots, H_k)$.
It then calls $\Sim{\route}$ from Theorem~\ref{thm:onetwork} on $H_1$.
After that the $\throw$ of the spilled elements is imitated via $\Sim{\throw}(H_1$, $H_2$, $H_{3}, \ldots, H_k)$
from Lemma~\ref{lemma:throw_with_dummy}.
It then again calls $\Sim{\route}$ but on its own version of $H_2$.
Once finished, on every search request from the adversary,
$\Sim{\search}$, that is not aware of the actions of $\Sim{\throw}$, is called.

If the algorithm does not fail during the insertion phase
and the simulator produces
valid zigzag paths during $\Sim{\throw}$ and $\Sim{\search}$, then the adversary obtains only
traces that are produced either by a hash function (i.e., a pseudo-random function)
or by a random function.
We showed that ZHT$_\throw$ and
$\Sim{\throw}$ fail with negligible probability.
The real search and $\Sim{\search}$ are also indistinguishable
even if the adversary controls the order of the inserted elements
as long as the adversary does not learn element's position in its zigzag path.
Since there is $\Sim{\route}$ which simulates $\route$,
the adversary does not learn element positions after routing.
Since every element participates in the throw, we showed in Lemma~\ref{lemma:throw_with_dummy}
that the adversary does not learn which elements are dummy.
Hence, the adversary does not learn positions of the elements across the tables
and ${\search}$ can be simulated successfully.
\end{proof}

\section{Pyramid ORAM}

We are now ready to present Pyramid ORAM by instantiating
the hierarchical construction in Section~\ref{sec:hierarchical}
with data-oblivious Zigzag hash tables.

\emph{Data Layout.}
Level $L_0$ is a list as before while every level $L_i$, $1 \le i \le l$,
is a Zigzag hash table with capacity $2^{i-1} \firstLevelSize$ and $k_i,c_i$ chosen
appropriately for Theorem~\ref{thm:ZHT} to hold.
We denote the bucket size and table number used for the last level with $c$ and $k$.
That is, $L_l$ is a ZHT$_{\N,k,c}$.

\emph{Setup.}
The setup is as described in Section~\ref{sec:hierarchical},
i.e., elements can be placed either in $L_l$ or inserted obliviously one by one
causing the initially empty ORAM to grow.

\emph{Access.}
The access to each level $i$ is an oblivious ZHT search, if the element
was not found in previous levels,
or a dummy ZHT access.

\emph{Rebuild.}
The rebuild phase involves calling the data-oblivious ZHT setup as described in
Section~\ref{sec:ozht}.
In particular, at level~$i$, a new ZHT for $2^{i-1}\firstLevelSize$
elements is set up.
Then, all elements from previous levels, including dummy,
participate in the initial throw.

We first determine the amortized rebuild cost of level~$i$ for an HORAM scheme
and then analyze the overall performance of \ourscheme{}.

\begin{lemma}
If the rebuild of level $i\ge 1$ takes $O(2^{i-1} \firstLevelSize C_i)$ steps, for $C_i> 0$,
then the amortized cost of the rebuild operation for every access to ORAM
is $\sum_{i=1}^{\log \N/\firstLevelSize} O(C_i)$.
\label{lemma:amortized_rebuild}
\end{lemma}
\begin{proof}
Since the rebuild of level $i$ happens every $2^{i-1}p$
accesses, the amortized cost of rebuilding level $i$
is $O(C_i)$.
Each access eventually causes a rebuild of every level, that is, $\log N/\firstLevelSize$ rebuilds
of different costs.
Hence, total amortized cost is:
$$
\sum_{i=1}^{\log \N/\firstLevelSize} O(C_i).
$$
\end{proof}
\begin{theorem}[Performance]
Pyramid ORAM incurs $O(k c \log \N)$
online and 
$O(k c (\log \N)^2)$
amortized bandwidth overheads,
$O(kc)$ space overhead,
and succeeds with very high probability.
\end{theorem}
\begin{proof}
The complexity of a lookup in a ZHT is $kc$.
The complexity of a lookup in the first level is $\firstLevelSize$.
There are $\log \N/\firstLevelSize$ levels.
So the cost of a lookup is $\firstLevelSize + kc \log \N/\firstLevelSize$.
Since $\firstLevelSize$ is usually set to $\log \N$, the
access time between table rebuilds, that is the online bandwidth overhead, is $O(kc \log \N)$.

The rebuild of level $i$ takes $O(2^{i-1} \firstLevelSize kc ((\log c)^2\log (2^{i-1} \firstLevelSize) + kc))$
steps
since $\alpha = kc$ (see~Corollary~\ref{corollary:dummy_real_zht}) and $k\ge k_i$ and $c\ge c_i$, for every~$i$.
Setting $C_i$ to $kc ((\log c)^2\log (2^{i-1} \firstLevelSize) + kc)$ and $\firstLevelSize$ to be at least~$2$
in Lemma~\ref{lemma:amortized_rebuild}, we obtain:
\begin{align*}
\sum_{i=1}^{\log \N/\firstLevelSize} & kc ((\log c)^2\log (2^{i-1} \firstLevelSize) + kc)
\\
&= kc \sum_{i=1}^{\log \N/\firstLevelSize} \left((\log c)^2\log (2^{i-1} \firstLevelSize) + kc\right)\\
&= kc \left((\log c)^2\sum_{i=1}^{\log \N/\firstLevelSize} (i-1 + \log \firstLevelSize)+  kc \log\N/\firstLevelSize \right)\\
&= kc \left((\log c)^2 (\log \N)^2 + kc \log\N \right) \\
&= O(kc (\log \N)^2 )
\end{align*}

Pyramid ORAM fails when the rebuild of some level~$i$
fails, that is, the number of tables allocated for level $i$ was not
sufficient to accommodate all real elements from the levels above.
Since, the number of real elements in tables at levels $0,1,\ldots, i-1$ is at most the
capacity of the ZHT at level $i$, ZHT fails with negligible probability as per analysis in
Theorem~\ref{thm:ozht}.
Hence, the overall failure probability of Pyramid ORAM can be bounded using a union bound.
\end{proof}

\begin{theorem}[Security]
Pyramid ORAM is data-oblivious with very high probability.
\end{theorem}
\begin{proof}
The adversary chooses $\N$ elements of its choice
and requests oblivious insertion in $L_l$. It then makes a polynomial number
of requests to elements in the ORAM.
We use $\Sim{\build}$ from Theorem~\ref{thm:oZHT}
and $\Sim{\search}$ from Lemma~\ref{lemma:ZHTOblivious1}
for every level except~$L_0$, which simulates naive scanning.

Between two rebuilds the search operation at each level is performed to distinct elements
and every table uses a new set of hash functions after its rebuild.
The hash functions are also different across the levels.
Hence, $\Sim{\search}$ can simulate the search at every level independently.

$\Sim{\build}$ for level $i$ takes as input the number of dummy
elements proportional to the total number of elements, real and dummy,
in $L_0, L_1, \ldots, L_{i-1}$.
We showed that we can construct~$\Sim{\build}$
even if the adversary knows the original
elements inserted in the table.
Hence, $\Sim{\build}$ succeeds when the adversary
controls the elements inserted in the ORAM as well as the access sequence.
\end{proof}

\section{Implementation} \label{sec:implementation}

We implemented Pyramid ORAM and Circuit ORAM~\cite{circuit} in C++
with oblivious primitives written in x64 assembly.
As ORAM block size, we use the x64 architecture's cache-line width of 64 bytes. Given that our assumed attacker is able to observe cache-line granular memory accesses, this block size facilitates the implementation of data-oblivious primitives.
A block comprises 56 bytes of data, the corresponding 4-byte index, and a 4-byte tag, which indicates the element's state (see Section~\ref{sec:network}).

To be \emph{data-oblivious}, our implementations are carefully crafted to be free of attacker-observable secret-dependent data or code accesses. We avoid secret-dependent branches on the assembly level either by employing conditional instructions or by ensuring that all branch's possible targets lie within the same 64-byte cache line. For data, we use registers as private memory and align data in memory with respect to cache-line boundaries.
Similar to previous work~\cite{secureml,raccoon}, we create a library of simple data-oblivious primitives using conditional move instructions and CPU registers. These allow us, for example, to obliviously compare and swap values.
Hash functions used in~\ourscheme{} are instantiated with our SIMD implementation of the Speck cipher.\footnote{\url{https://en.wikipedia.org/wiki/Speck_(cipher)}}

\subsection{Pyramid ORAM}
For bucket sizes of $c > 4$, our implementation of Pyramid ORAM uses a Batcher's sorting network for the re-partitioning of two buckets within our \emph{probabilistic routing network} described in Section~\ref{sec:network}.

For $c = 4$, we employ a custom re-partition algorithm that is optimized for the x64 platform and leverages SIMD operations on 256-bit AVX2 registers. The implementation is outlined in the following.
For most operations, we use AVX2 registers as vectors of four 32-bit components. To re-partition two buckets $A$ and $B$, we maintain the tag vectors $\mathbf{t}_A$ and $\mathbf{t}_B$ and position vectors $\mathbf{p}_A$ and $\mathbf{p}_B$. Initially, we load the four tags of bucket $A$ into $\mathbf{t}_A$ and those of bucket $B$ into $\mathbf{t}_B$. A tag indicates an element's state.
Empty elements and elements that were spilled in the previous
stages of the network are assigned -1; elements that
have not been spilled and need to be routed to either bucket $A$ or $B$
are assigned 0 and 1, respectively.

$\mathbf{p}_A$ and $\mathbf{p}_B$ map element positions in the input buckets to element positions in the output buckets.
We refer to the elements in the input buckets by indices. The elements in input bucket $A$ are assigned the indices $1 \ldots 4$, those in input bucket $B$ the indices $5 \ldots 8$. Correspondingly, we initially set $\mathbf{p}_A = (1,2,3,4)$ and $\mathbf{p}_B = (5,6,7,8)$. 

We then sort the tags by performing four oblivious SIMD compare-and-swap
operations between $\mathbf{t}_A$ and $\mathbf{t}_B$. Each such operation comprises four 1-to-1 comparisons. (Hence, a total of 16 1-to-1 compare-swap-operations are performed.) Every time elements are swapped between $\mathbf{t}_A$ and $\mathbf{t}_B$, we also swap the corresponding elements in $\mathbf{p}_A$ and $\mathbf{p}_B$.
In the end, $\mathbf{t}_A$ and $\mathbf{t}_B$ are sorted and $\mathbf{p}_A$ and $\mathbf{p}_B$ contain the corresponding mapping of input to output positions.
Finally, we load the two input buckets entirely into the 16 available AVX2 registers (which provide precisely the required amount of private memory) and obliviously write their elements to the output buckets according to $\mathbf{p}_A$ and $\mathbf{p}_B$.

\paragraph{Instantiation}
In our implementation, the ORAM's first level holds 1024 cache-lines, because for smaller array sizes scanning performs better.
All other levels use buckets of size $c=4$ and the number of tables
is set to $\lceil \log \log \rceil$ of the capacity of the level.
We empirically verified parameters $k=c=4$ for tables of size $2^{11}$ and $2^{15}$
when re-throwing
elements $2^{23}$ and $2^{20}$ times, respectively.
In both cases, we
have not observed spills into the third table (in any of the levels), hence, the last two tables of every level were always empty.

\subsection{Circuit ORAM}
For comparison, we implemented a recursive version of Circuit ORAM~\cite{circuit} with a deterministic eviction strategy.
Circuit ORAM proposes a novel protocol to implement
the eviction of tree-based ORAMs using a constant amount of private memory. This makes Circuit ORAM the state-of-the-art tree-based ORAM protocol with optimal circuit complexity for $O(1)$ client memory.
Similar to other tree-based schemes Circuit ORAM stores elements in a binary tree
and maintains a mapping between element indices and their position in the tree
in a recursive position map.
We refer the reader to~\cite{circuit} for the details of Circuit ORAM.

We implemented Circuit ORAM from scratch using our library of primitives, as existing implementations\footnote{\url{https://github.com/wangxiao1254/oram_simulator}} were built for simulation purposes and only concern with high-level obliviousness.
This makes these implementations not fully oblivious and vulnerable to low-level side-channel attacks based on control flow.
Similar to \ourscheme{}, we use a cache line as a block for Circuit ORAM.
This favors Circuit ORAM since
up to 14 indices can be packed in each block, requiring only $\log_{14} \N$ levels of recursion
to store the position map.
Hence, the bandwidth cost of our implementation of Circuit ORAM is only 64 bytes $ \times~O(\log_{14}\N \log \N)$.
Similar to Pyramid ORAM, scanning is faster than accessing binary trees
for small datasets.
To this end, we stop the recursion
when the position map is reduced to 512 indices or less
and store these elements in an array that is scanned on every access.
We chose 512 since it gave the best performance.

\paragraph{Instantiation}
We use a bucket size of $3$ with two stash parameters: 12 and 20, that is every level of recursion uses
the same stash size since these parameters are independent of the size of a tree,
and a bucket size of $4$ and stash size 12.
We note that buckets of size~2 were not sufficient in our experiments to avoid overflows.
In the figures we denote the three settings with Circuit-3-12, Circuit-3-20 and Circuit-4-12.

\section{Evaluation} \label{sec:evaluation}
We experimentally evaluate the performance of Pyramid ORAM in comparison to Circuit ORAM and naive scanning when accessing arrays of different sizes, when used in
two machine learning applications, and when run inside Intel SGX enclaves. The key insight is that, for the given datasets, Pyramid ORAM considerably outperforms Circuit ORAM (orders of magnitude) in terms of \emph{online} access times and is still several \textit{x} faster for 99.9\% of accesses when rebuilds are included in the timings.

All experiments were performed on
SGX-enabled Intel Core i7-6700 CPUs running Windows Server 2016. All code was compiled with the Microsoft C/C++ compiler v19.00 with \texttt{/O2}. In cases where we ran code inside SGX enclaves, it was linked against a custom SGX software stack. 
We report all measurements in CPU cycles.

\paragraph{ORAM vs.~Scanning}
Figure~\ref{fig:all_compare_14} shows the break-even point of Pyramid and Circuit ORAM compared to the basic approach---called \emph{Naive} in the following---of scanning through all elements of an array for each access. Our implementation of Naive is optimized using AVX2 SIMD operations and is the same as the one we use to scan the first  level in our Pyramid implementation.
Hence, for up to 1024 elements (the size of Pyramid's first level), Naive and Pyramid 
perform identically. We also set Circuit ORAM to use scanning for small arrays.
Pyramid and Circuit start improving over Naive at $4096$ (229 KB) and $8192$ elements (459 KB), respectively.

\begin{figure}[t]
\centering
    \includegraphics[width=\columnwidth]{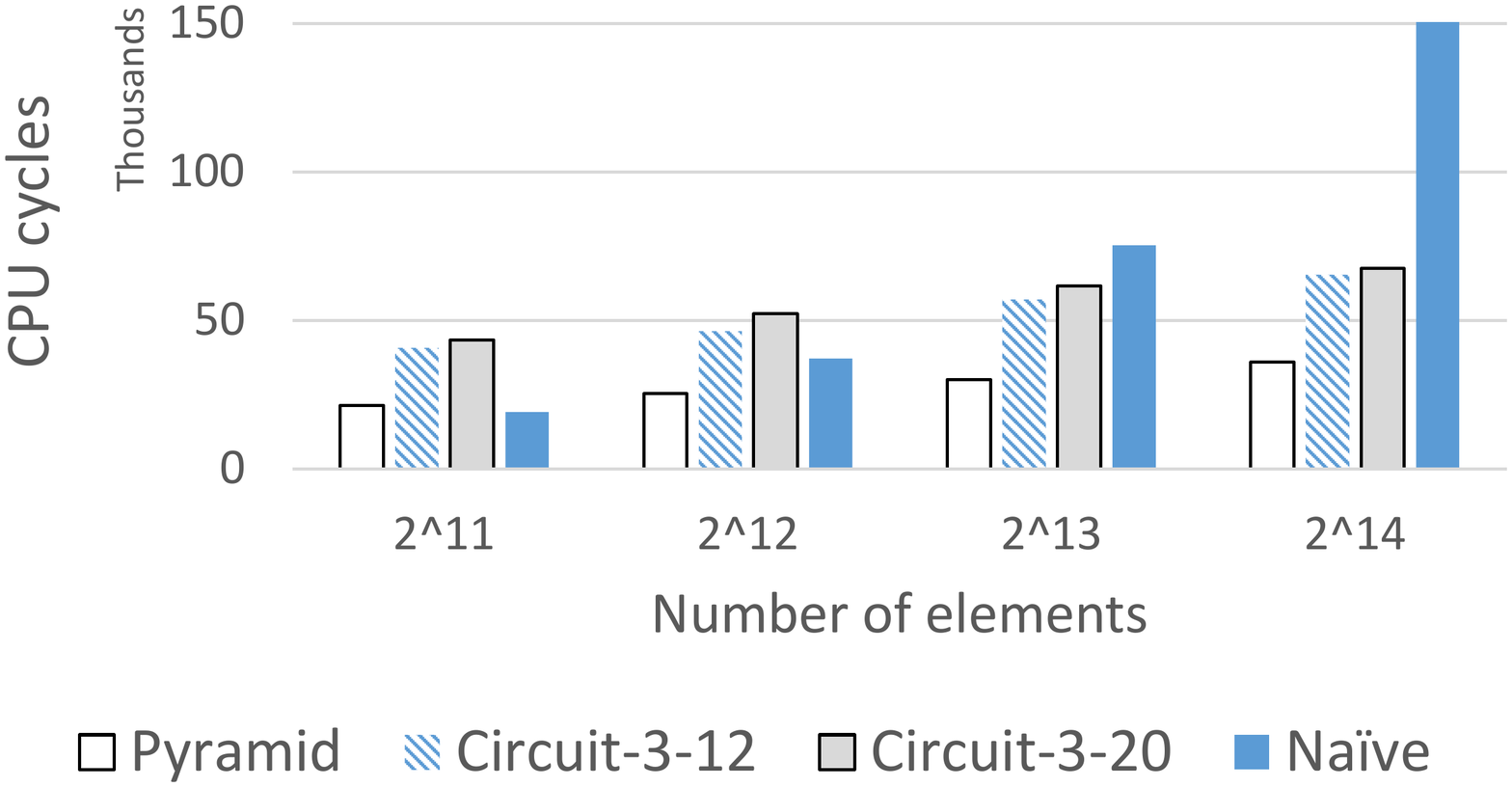}
\caption{Average overhead of obliviously reading elements from an array of 56-byte elements. Array sizes vary from $2048$ to~$16384$ elements (114KB--917KB).}
\label{fig:all_compare_14}
\end{figure}

\paragraph{Pyramid vs.~Circuit}
Figures~\ref{fig:boolean} and~\ref{fig:custom} depict the performance of Pyramid and Circuit
when sequentially accessing arrays of 1-byte and 56-byte elements.
Pyramid dominates Circuit for datasets of size up to 14~MB.
Though Circuit ORAM should perform worse for smaller
elements than larger ones, it is not evident in the 1-byte and 56-byte experiments
since our implementation uses (machine-optimized) large block sizes 
favoring Circuit ORAM.

\begin{figure}[t]
\centering
    \includegraphics[width=\columnwidth]{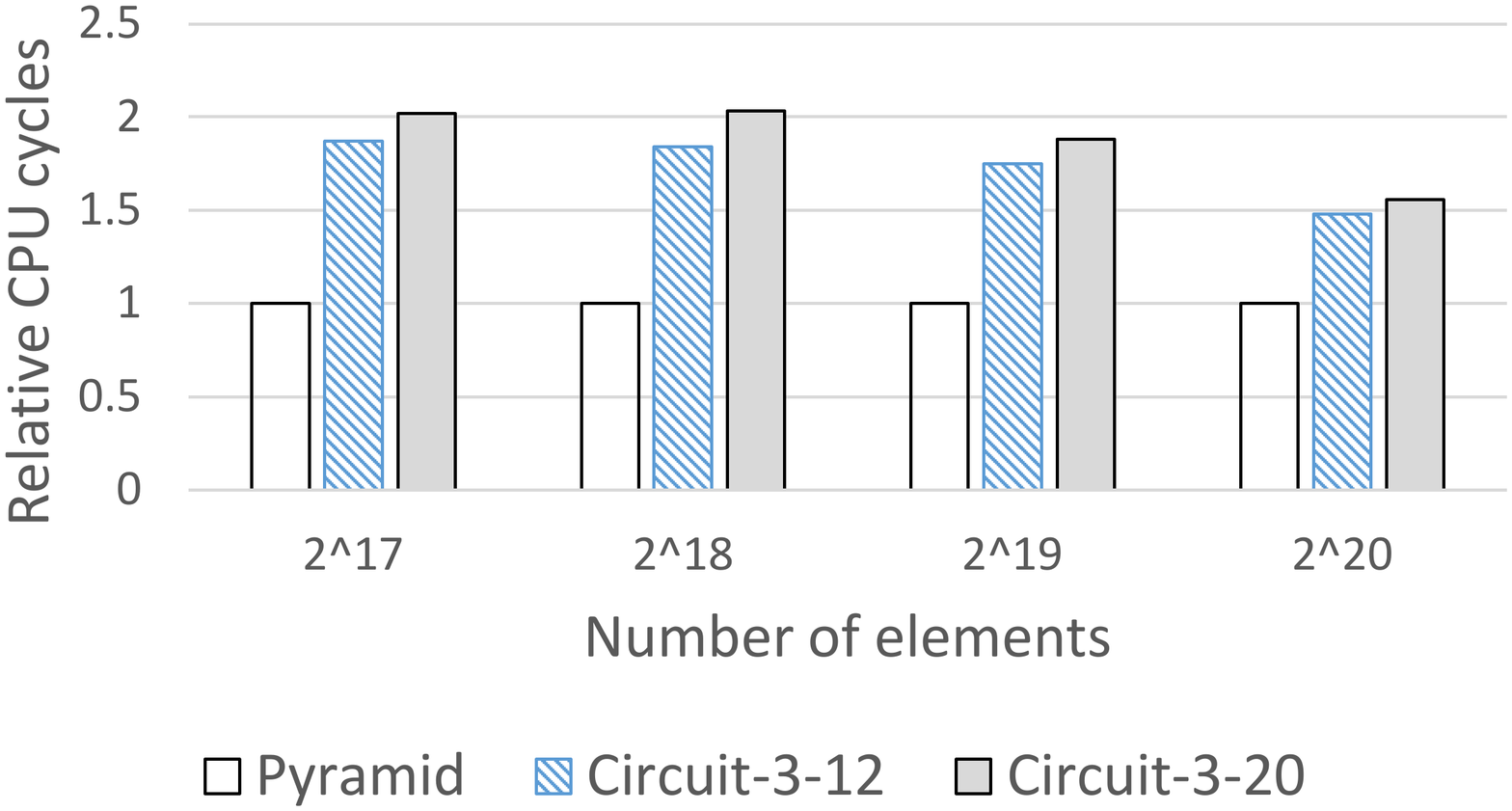}
\caption{
Average CPU cycle overhead of obliviously reading 1-byte elements from arrays of size  $2^{17}$ to $2^{20}$ (4KB--1MB) in an SGX enclave;
the measurements are normalized to the performance of~\ourscheme{}.}
\label{fig:sgx_boolean}
\end{figure}

\begin{figure}[t]
\centering
    \includegraphics[width=\columnwidth]{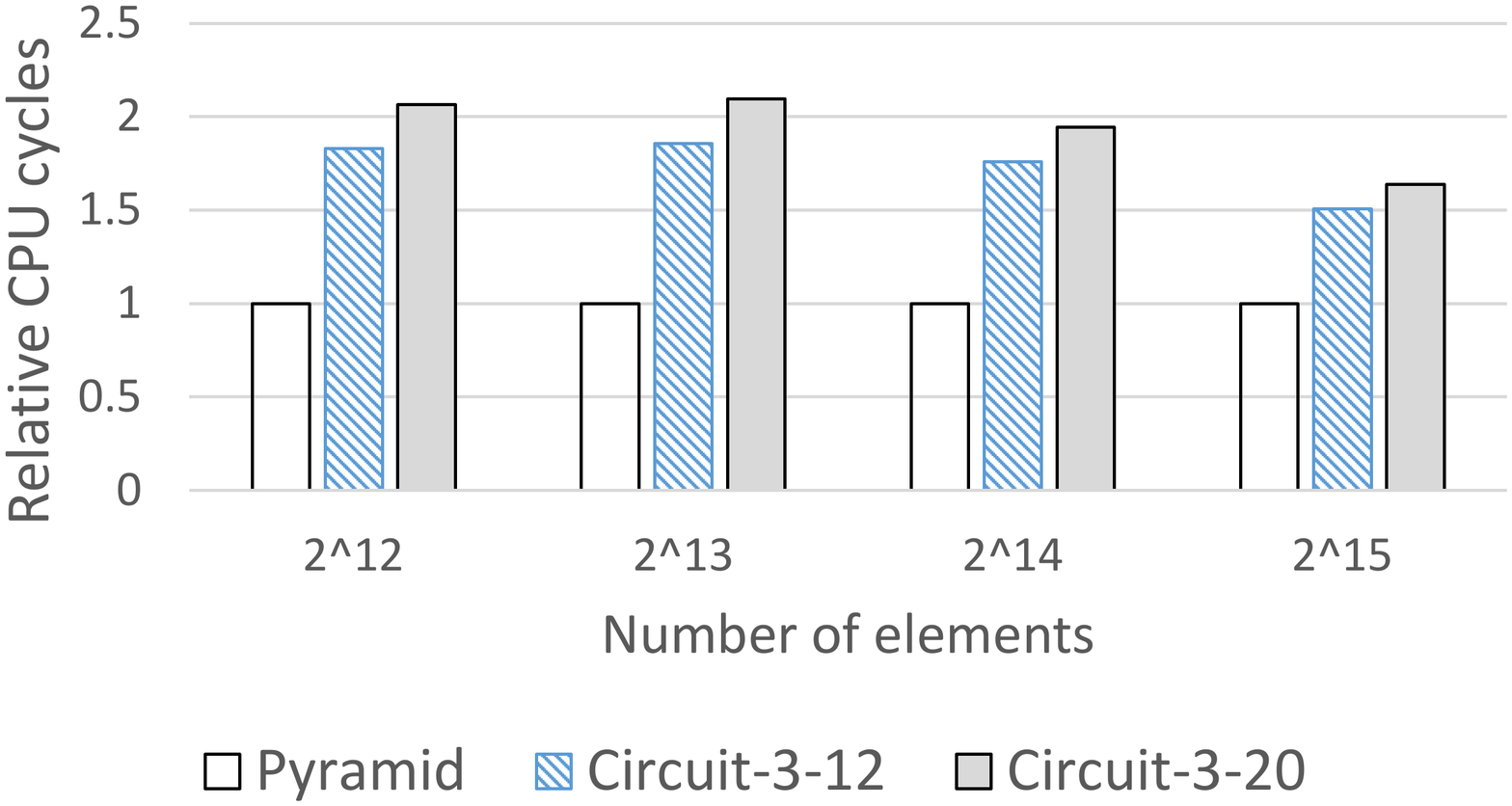}
\caption{Average CPU cycle overhead (normalized to \ourscheme{}) of obliviously reading 56-byte elements from arrays of size $4096$ to $32768$ (229KB--1.8MB)
{from an Intel SGX enclave}.
}
\label{fig:sgx_custom}
\end{figure}

\begin{figure}[t]
\centering
    \includegraphics[width=\columnwidth]{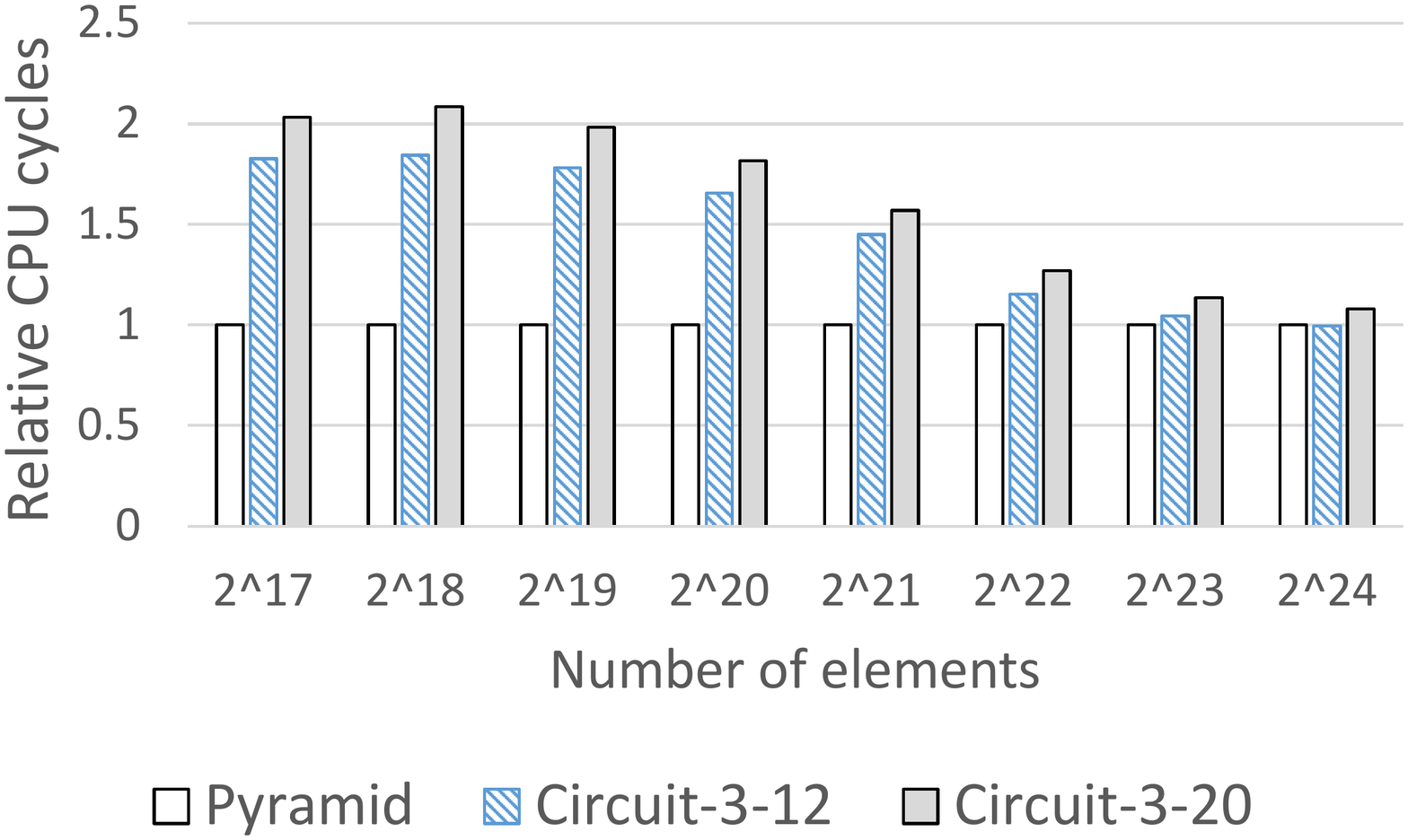}
\caption{Average CPU cycle overhead (normalized to \ourscheme{}) of obliviously reading 1-byte elements from arrays of size $2^{17}$ to $2^{24}$ (131KB--16MB).}
\label{fig:boolean}
\end{figure}

\begin{figure}[t]
\centering
    \includegraphics[width=\columnwidth]{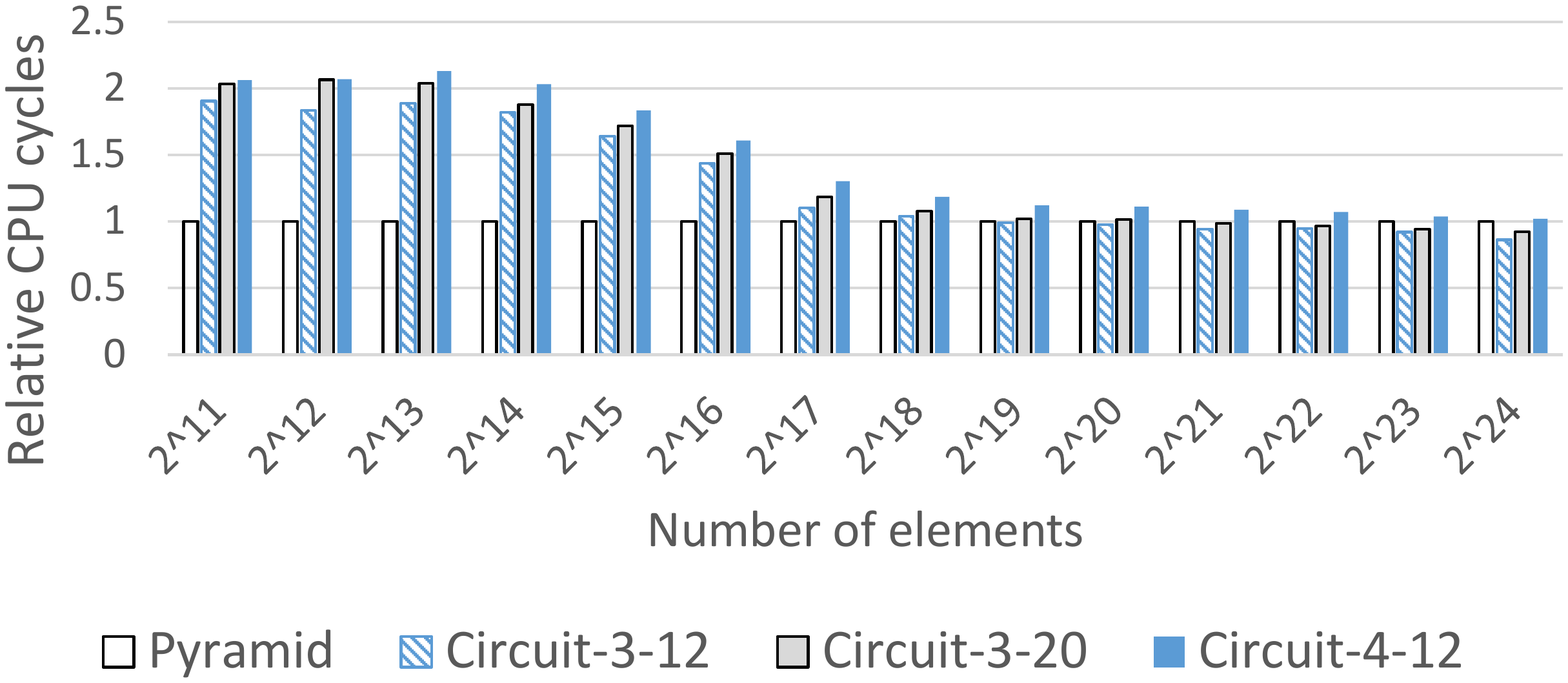}
\caption{
Average CPU cycle overhead (normalized to \ourscheme{}) of obliviously reading 56-byte elements from arrays of size {$2^{11}$ to $2^{24}$} (114KB--939MB)}
\label{fig:custom}
\end{figure}

\paragraph{Intel SGX}
We run several benchmarks inside of an SGX enclave.
Due to memory restrictions of around 100 MB in the first generation of Intel SGX,
we evaluate the algorithms on small datasets.
Since both Circuit and Pyramid have space overhead, the actual data that can be loaded in an ORAM is at most 12 MB.
Compared to the results in the previous section, SGX does not add any
noticeable overhead. Figures~\ref{fig:sgx_boolean} and~\ref{fig:sgx_custom} show the results for arrays with 1-byte and 56-byte
elements.

\paragraph{Variation of access overhead}
As noted earlier, the overhead of accessing a hierarchical ORAM depends on the number of previous accesses.
For example, an access that causes a rebuild of the last level is
much slower than the subsequent access (for which no rebuild happens).
Hence, \ourscheme{} exhibits different online and amortized cost as noted in Table~\ref{tbl:cmp}.
To investigate this effect, we accessed Pyramid and Circuit $3\N$ times and measured the
overhead of every request for $\N = 2^{20}$ and $\N=2^{25}$.
For both schemes the measured time includes the time to access the element
and to prepare the data structure for the next request,
that is, it includes rebuild and eviction for Pyramid ORAM and Circuit ORAM, respectively.
Figures~\ref{fig:single_accesses_scale} and~\ref{fig:single_accesses_scale_25} show the corresponding overhead split across 99.9\% of the fastest among $3\N$ accesses. As can be seen, the fraction of expensive accesses in Pyramid is very small since large rebuilds are infrequent.
For $\N= 2^{20}$, almost $40\%$ of all Pyramid accesses are answered on average in 16.6K CPU cycles,
almost 10$\times$ faster than Circuit.
Except for the few accesses that cause a rebuild,
Pyramid answers $99.9\%$ of all accesses on average in 26.8K CPU cycles,
compared to 178.7K cycles for Circuit.
A similar trend can be observed for the larger array in Figure~\ref{fig:single_accesses_scale_25}:
for $35\%$ of all accesses, Pyramid is at least 10$\times$ faster than Circuit
and still 8$\times$ faster for $99\%$ of all accesses,
while, on average, Circuit performs better.
We stress that for many applications the online cost is the significant characteristic, as expensive rebuilds can be anticipated and performed during idle times or, to some extent, can be performed concurrently.

\begin{figure}[t!]
\centering
      \includegraphics[width=\columnwidth]{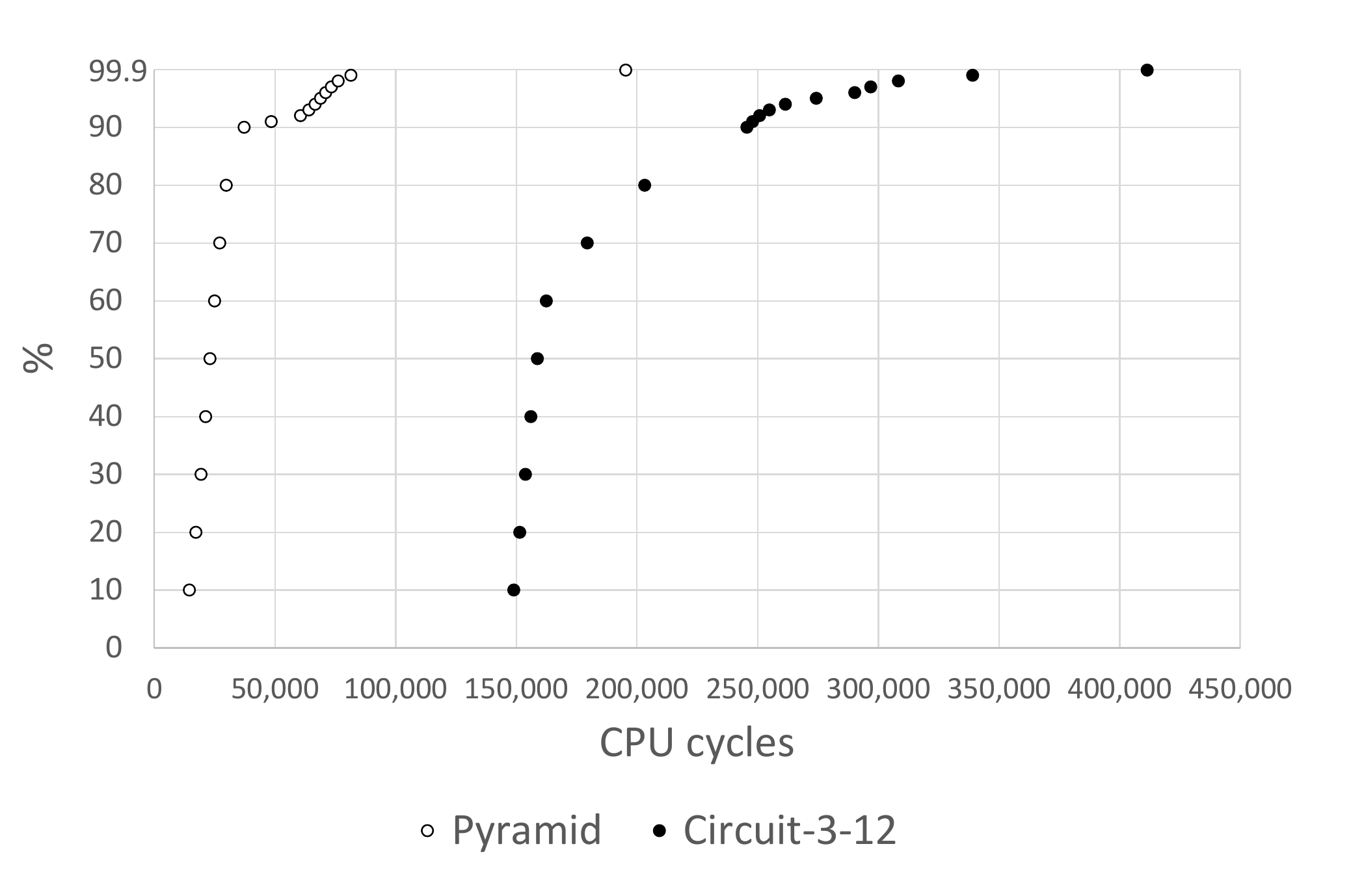}
\caption{Cumulative distribution of access times in CPU cycles
when requesting elements from an array of $2^{20}$ 56-byte elements (58 MB)
where 99.9\% of fastest accesses are plotted.
For example, Pyramid ORAM (Circuit ORAM) answers 50\% and~99\% of all requests under~23K (159K)
and~82K (340K) CPU cycles.}
\label{fig:single_accesses_scale}
\end{figure}

\begin{figure}[t!]
\centering
      \includegraphics[width=\columnwidth]{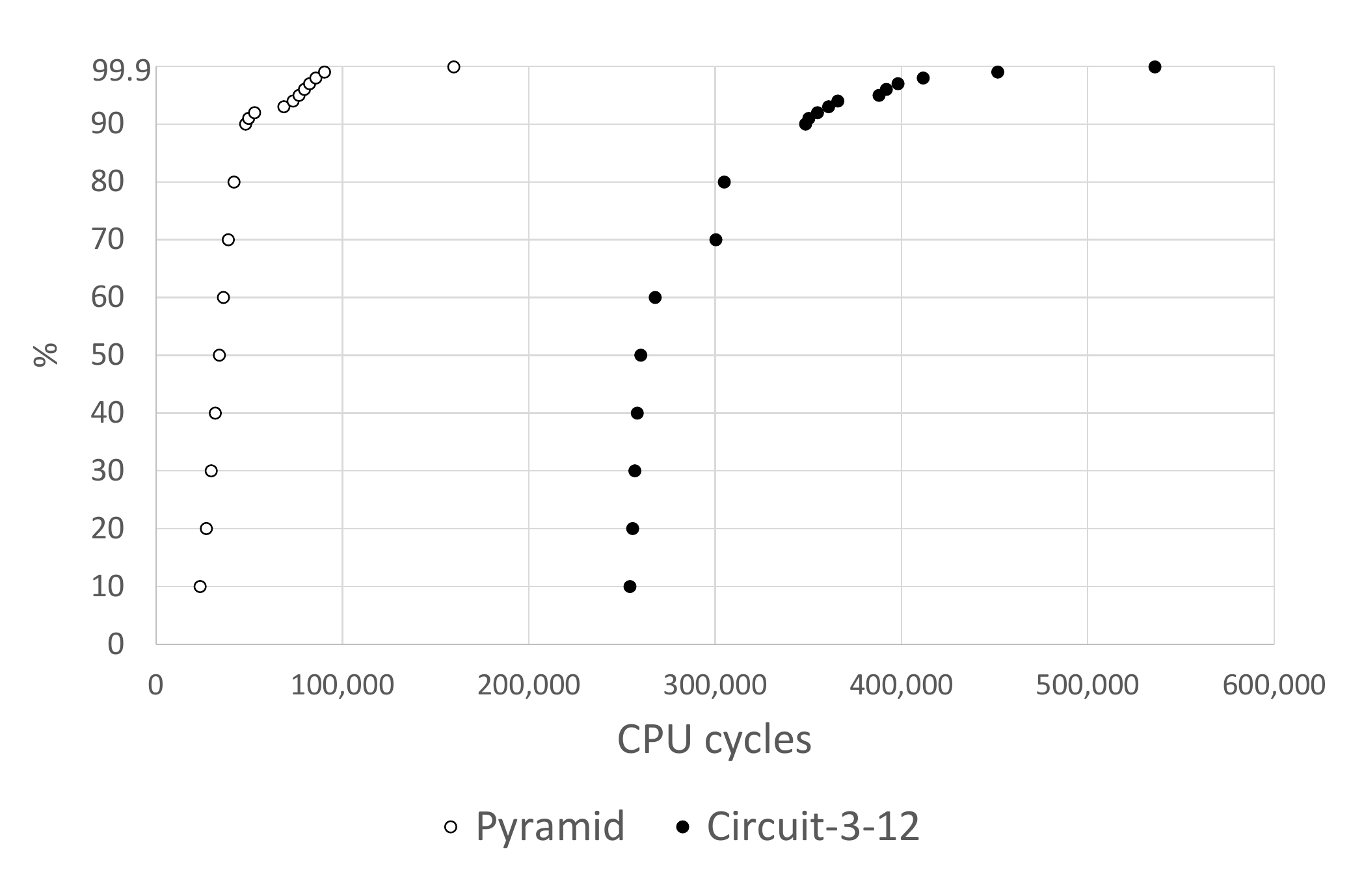}
\caption{Cumulative distribution of access times in CPU cycles
when requesting elements from an array of $2^{25}$ 56-byte elements (1.87GB)
where 99.9\% of fastest accesses are plotted.
For example, Pyramid ORAM (Circuit ORAM) answers 50\% and~99\% of all requests under~33.8K (260K)
and~90K (452K) CPU cycles.}
\label{fig:single_accesses_scale_25}
\end{figure}

\paragraph{Sparse Vectors}
Dot product between a vector and a sparse vector
is a common operation
in machine-learning algorithms.
For example, in linear SVM~\cite{KeerthiD05},
a model contains a weight for every feature while
a sample record may contain only some of the features.
Since relatively few features may be present in a sample record,
records are stored as
(feature index, value) pairs.
The dot product is computed by performing random accesses into the model
based on the feature indices, in turn, revealing
the identities of the features present in the record.

We use ORAM to store and access the model vector
of 4-byte floating feature values in
order to protect record content during classification.
To measure the overhead, we choose three datasets with sparse sample records from~LIBSVM Tools\footnote{\url{https://www.csie.ntu.edu.tw/~cjlin/libsvmtools/datasets/} (accessed 15/05/2017).}:
\texttt{news20.binary} (1,355,191 features and 19,996 records), \texttt{rcv1-train} (47,236 features and 20,242 records)
and \texttt{rcv1-test} (47,236 features and 677,399 records).
These datasets are common datasets for text classification
where a feature is a binary term signifying the presence of
a particular trigram.
As a result the datasets are sparse (e.g., if the first dataset were to be stored
densely, only $0.0336$\% of its values would be non-zero).
In Figure~\ref{fig:sparse_vector} we compare the overhead
of the ORAM schemes over insecure baseline when computing dot product on each dataset.
The insecure baseline measures sparse vector multiplication
without side-channel protection.
As expected, Pyramid performs better than Circuit,
as model sizes are 5.4MB for \texttt{news20.binary}
and 188KB for \texttt{rcv1}.
Comparing to the baseline, the average ORAM overhead is four and
three orders of magnitude, respectively.

\begin{figure}[t]
\centering
    \includegraphics[width=\columnwidth]{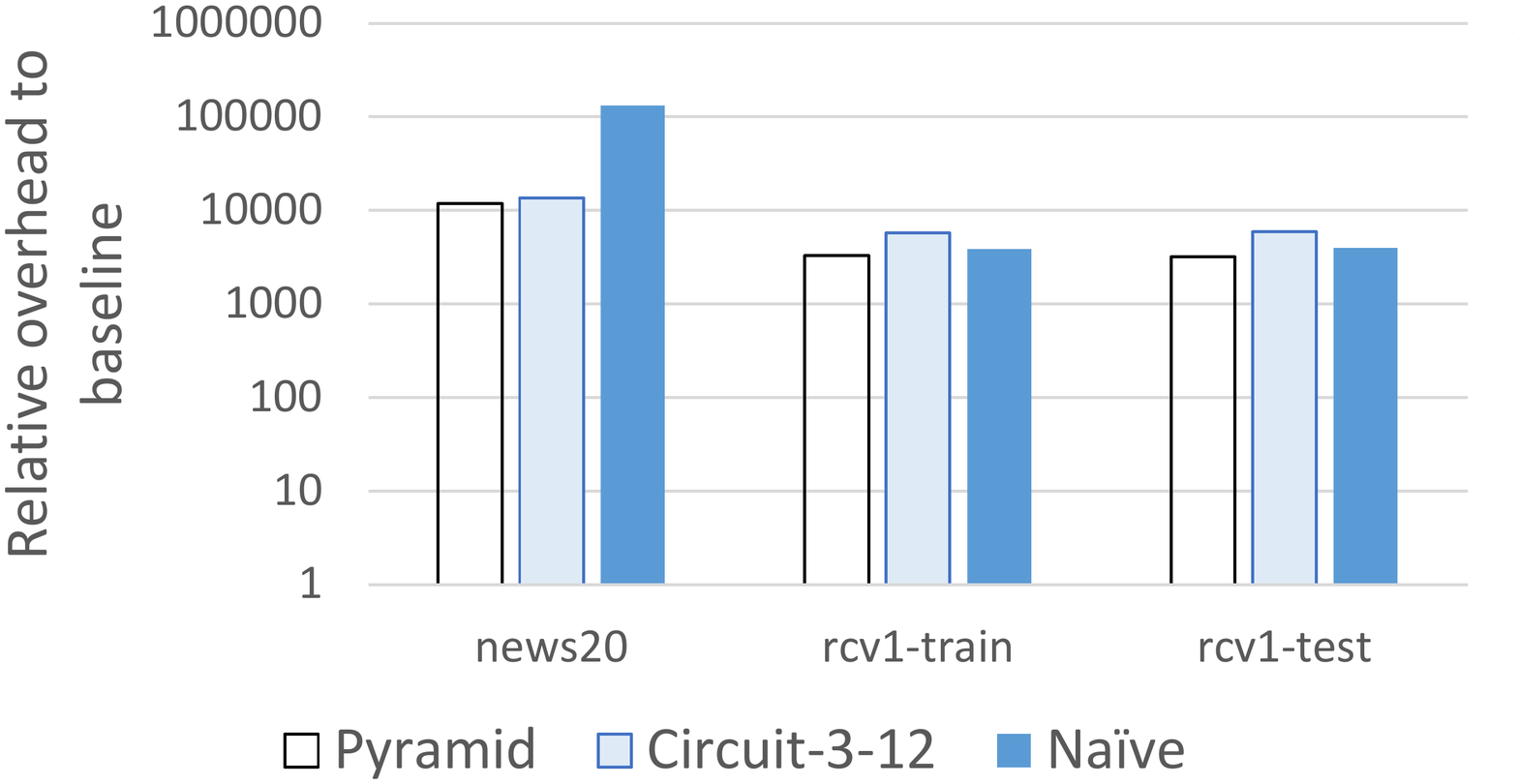}
\caption{Sparse vector multiplication: CPU cycle overhead of ORAM solutions on three datasets normalized to the performance
of the insecure baseline running on the same dataset.}
\label{fig:sparse_vector}
\end{figure}

\vspace{-10pt}

\paragraph{Decision Trees}
Tree traversal is also a common task in machine learning
where a model is a set of decision trees (a forest)
and a classification on a sample record is done by traversing the trees
according to record's features. Hence, accesses to the tree
can reveal information about sample records.
One can use an ORAM to protect such data-dependent accesses
either by placing the whole tree in an ORAM or by placing each layer of the tree in a separate ORAM.
The latter approach suits balanced trees, however, for unbalanced trees it reveals the height of the tree
and the number of nodes at each layer.

Our experiments use a pre-trained forest for a large public dataset\footnote{The \texttt{Covertype} dataset is available from the UCI Machine Learning Repository: \url{https://archive.ics.uci.edu/ml/datasets/covertype/} (accessed 18/05/2017)}. The forest consists of 32 unbalanced trees with 30K to 33K nodes each. A node comprises four 8-byte values and hence fits into a single 56-byte block. We load each tree into a separate ORAM and evaluate 290,506 sample records; the performance results are given in Figure~\ref{fig:trees}.
Pyramid's low online cost shows again: Pyramid and Circuit incur three orders of magnitude overhead on average, while $99\%$ of all classifications on average have only two orders of magnitude overhead for Pyramid.

In summary, Pyramid ORAM can give orders of magnitude lower access latencies than Circuit ORAM. In particular, Pyramid's online access latencies are low. This makes Pyramid the better choice for many applications in practice.

\begin{figure}[t]
\centering
    \includegraphics[width=.85\columnwidth]{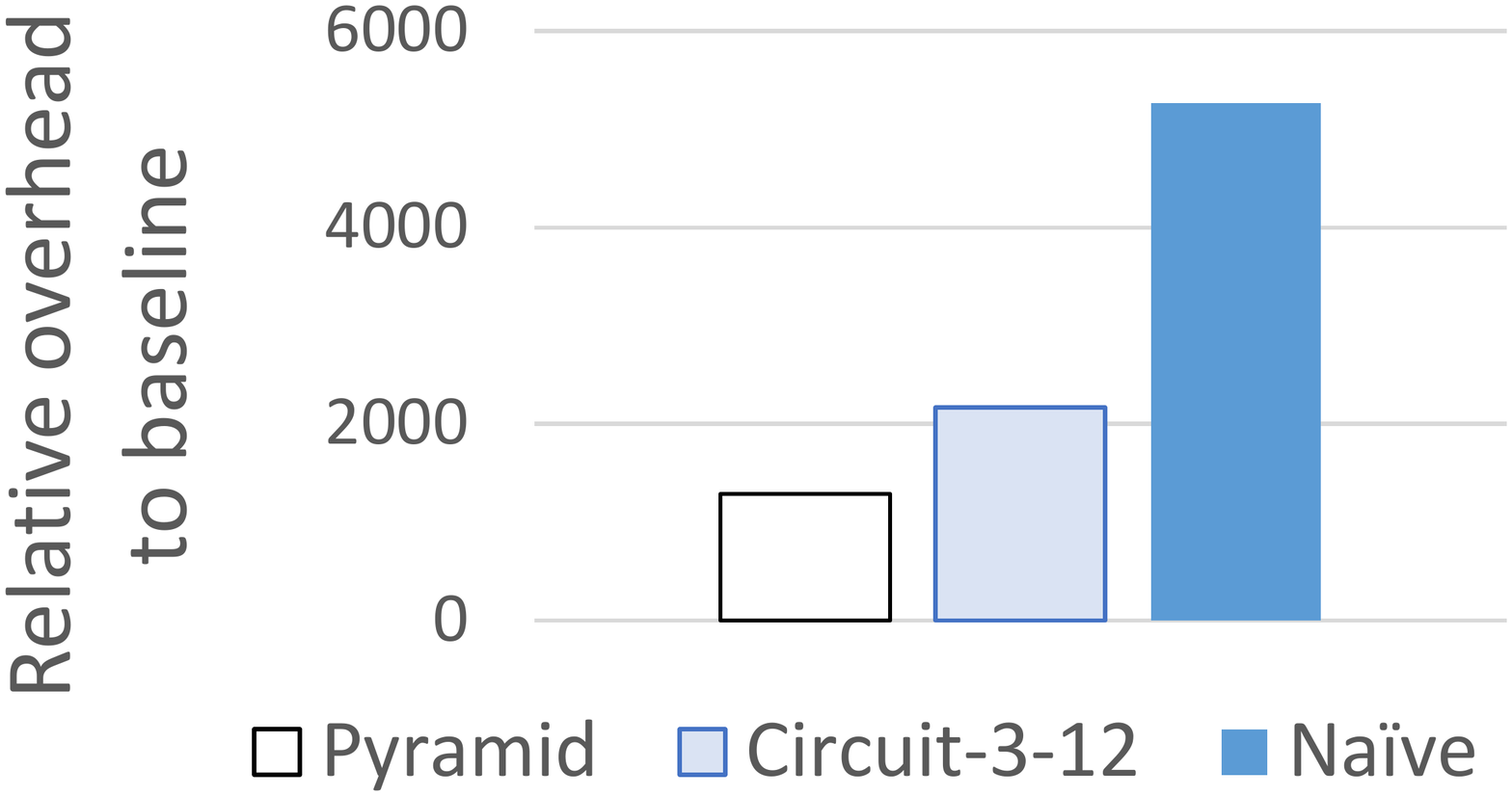}
\caption{Decision trees: CPU cycle overhead of ORAM solutions normalized to the performance of the insecure baseline.}
\label{fig:trees}
\end{figure}

\section{Related work} \label{sec:relatedwork}

\paragraph{Oblivious RAM}
The hierarchical constructions were the first instantiations proposed
for oblivious RAM~\cite{GoldreichO96}.
Since then, many variations have been proposed which alter
hash tables, shuffling techniques, or assumptions on the size of private memory~\cite{Pinkas2010,GMOT12,wsc-bcomp-08,gm-paodor-11}.
The hierarchical construction
by Kushilevitz~\textit{et al.}~\cite{KLO}
has the best asymptotical performance
among ORAMs for constant private memory size.
It uses the oblivious cuckoo hashing scheme proposed by Goodrich and Mitzenmacher~\cite{gm-paodor-11}.
In particular, the scheme relies on the AKS sorting network~\cite{AKS},
which is performed $O(\log \N)$ times
for every level in the hierarchy.
The AKS network, though an optimal sorting network, in simplified form has a high hidden constant of 6,100~\cite{DBLP:journals/algorithmica/Paterson90}.
Moreover, until the hierarchy reaches
the level that can hold at least $(\log \N)^7$ elements,
every level is implemented using
regular hash tables with buckets of size $\log \N$.
Thus, AKS only becomes practical for large data sizes.

Tree-based ORAM schemes, first proposed by Shi~\textit{et al.}~\cite{tree_based_orams}, are a fundamentally different paradigm. There have been a number of follow-up constructions~\cite{pathoram, ringoram, burstoram} that improve on the performance of ORAM schemes. Most of these works focus on reducing bandwidth overhead, often relying on non-constant private memory and elements of size $\Omega(\log \N)$ or $\Omega((\log \N)^2)$.
We note that such schemes can be modified to rely
on constant private memory by storing its content in an encrypted form externally
and accessing it also obliviously (e.g., a scanning for every access).
However, this transformation
is often expensive. For example, it increases the asymptotical overhead of Path ORAM, because stash operations now require using oblivious sorting~\cite{scoram,circuit}.

ORAMs have long held the promise to serve as building blocks for multi-party computation (MPC). Hence some recent works focus on developing protocols with small circuit complexity~\cite{scoram,7546504}, including the Circuit ORAM scheme~\cite{circuit} discussed in previous sections.
Some of the schemes that try to decrease circuit size also have small private memory requirements,
since any ``interesting'' operation on private memory (e.g., sorting) increases the size.
However, not all ORAM designs that target MPC can be used in our
client-server model. For example, the work by Zahur~\textit{et al.}~\cite{7546504}
optimizes a square-root ORAM~\cite{GoldreichO96} assuming that each party can
store a permutation of the elements.

\paragraph{Memory Side-channel Protection for Hardware}
Rane~\textit{et al.}~\cite{raccoon} investigated compiler-based
techniques for removing secret-dependent accesses in the code running on commodity hardware.
Experimentally, they showed
that Path ORAM is not suitable in this setting due to its requirements on
large private memory, showing that naive scanning performs better for
small datasets.

Ring ORAM~\cite{ringoram} optimizes parameters of Path ORAM and uses
an XOR technique on the server to decrease the bandwidth cost.
As a case study, Ren~\textit{et al.} estimate the performance of Ring ORAM
if used on secure processors assuming 64-byte element size. The design
assumes non-constant size on-chip memory that can fit the stash,
a path of the tree (logarithmic in $\N$) and the position map at a certain recursion level.
Ascend~\cite{ascend} is a design for a secure processor
that accesses code and data through Path ORAM.
As a result, Ascend relies on on-chip private memory
that is large enough to store the stash and a path internally (i.e., $O(\log \N)$).
PHANTOM~\cite{phantom} is a processor design with an oblivious
memory controller that also relies on Path ORAM, storing the stash in trusted memory
and using 1KB--4KB elements. Fletcher~\textit{et al.}~\cite{Fletcher:2015:LLH:2860695.2860768} improve
on this using smaller elements, but they also assume large on-chip private memory.
Other hardware designs optimized for Path ORAM exist~\cite{Ren:2013:DSE:2485922.2485971,Fletcher:2015:FON:2694344.2694353, flatoram}.
Sanctum~\cite{sanctum} is a secure processor design related to Intel SGX. Other than SGX, it provides private cache partitions for enclaves, stores enclave page-tables in enclaves, and dispatches corresponding page faults and other exceptions directly to in-enclave handlers. This largely prevent malicious software (including the OS) from inferring enclave memory access patterns. However, it does not provide additional protection against a hardware attacker. 

Recent work also addresses SGX's side channel problem in software: Shih~\textit{et al.}~\cite{tsgx} proposed executing sensitive code within Intel TSX transactions in order to suppress and detect leaky page faults. Gruss~\textit{et al.}~\cite{cloak} preload sensitive code and data within TSX transactions in order to prevent leaky cache misses. 
Ohrimenko~\textit{et al.}~\cite{secureml} manually protect a set of machine learning algorithms against memory-based side-channel attacks, applying high-level algorithmic changes and low-level oblivious primitives. 
DR.SGX~\cite{drsgx} aims to protect enclaves against cache side-channel attacks by frequently re-randomizing data locations (except for small stack-based data) at cache-line granularity. For this, DR.SGX augments a program's memory accesses at compile time. DR.SGX incurs a reported \emph{mean} overhead of 3.39x--9.33x, depending on re-randomization frequency. DR.SGX does not provide strong guarantees; it can be seen as an ad-hoc ORAM construction that may induce enough noise to protect against certain software attacks.
ZeroTrace~\cite{zerotrace} is an oblivious ``memory controller'' for enclaves. It implements Path ORAM and Circuit ORAM to obliviously access data stored in untrusted memory and on disk; only the client's position map and stash are kept in enclave memory and are accessed obliviously using scanning and conditional x86 instructions. 

\paragraph{Routing and Hash Tables}
The probabilistic routing network in Section~\ref{sec:network}
belongs to the class of
unbuffered delta networks~\cite{Patel81}.
The functionality of our network is different as instead of dropping
the elements that could not be routed, we label them
as spilled (i.e., tagged $\false$) and carry them to the output.
Probably the closest to ours, is the network based on
switches with multiple links studied by Kumar~\cite{network_thesis}.
The network assumes that each switch of the first stage (input buckets in our case) receives one packet,
which is not the case in our work where buckets may contain up to $c$ elements.
Moreover, the corresponding analysis, similar to other analyses of delta networks,
focuses on approximating the throughput of the network, while we are interested in the
upper-bound of the number on
elements that are not routed.

The Zigzag hash table can be seen as a combination of the load balancing strategy called
$\mathsf{Threshold}(c)$~\cite{Adler:1995:PRL:225058.225131}
and Multilevel Adaptive Hashing~\cite{hashtable}.
In the former, elements spilled from a table with buckets of size $c$
are re-thrown into the same table, discarding
the elements from the previous throws. 
In Multilevel Adaptive Hashing, $k$ tables of different sizes with buckets of size 1 are used
to store elements. The insertion procedure is the same as that for \emph{non-oblivious} ZHT,
that is elements that collide in the first hash table are re-inserted into the next table.
Different from ZHTs, multilevel hashing \emph{non-obliviously} rebuilds its tables when the $k$th table overflows.

\section{Conclusions} \label{sec:conclusions}
We presented \ourscheme{}, a novel ORAM design 
that targets modern SGX processors.
Rigorous analysis of \ourscheme{}, as
well as an evaluation of an optimized implementation
for x64 Skylake CPUs, shows that \ourscheme{} provides strong
guarantees and practical performance while outperforming
previous ORAM designs in this setting either in asymptotical online complexity or in practice.

\section*{Acknowledgments}
We thank C\'edric Fournet, Ian Kash, and Markulf Kohlweiss for helpful discussions.

\balance
\bibliographystyle{plain}
\bibliography{main}

\appendix

\begin{figure*}[t!]
\centering
    \includegraphics[scale=0.39]{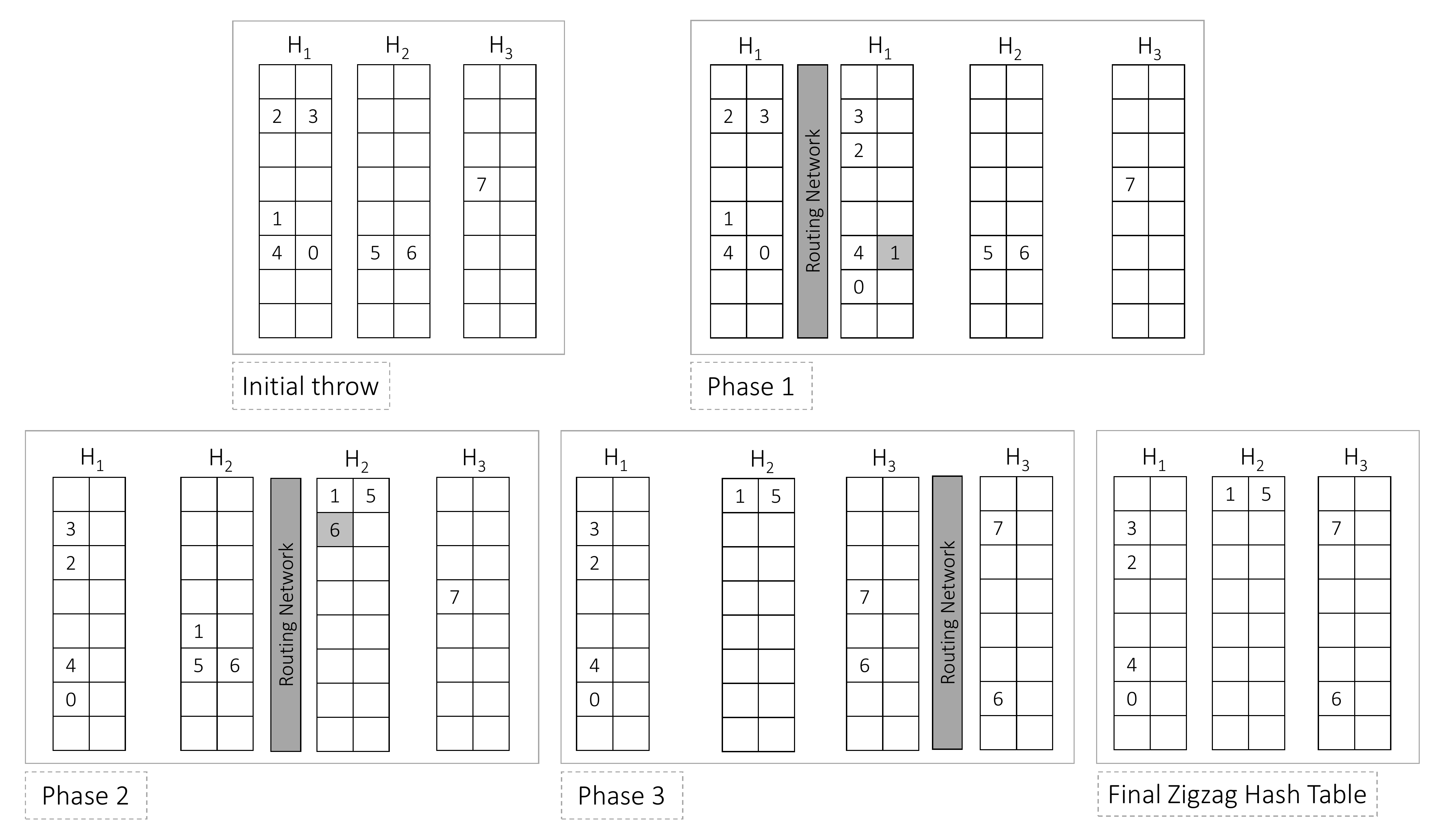}
\caption{Oblivious Zigzag Hash Table Setup from Section~\ref{sec:ozht}.
During the initial throw elements are allocated using random hash functions.
In phase~1 elements in $H_1$ are routed with probabilistic routing network
using hash function~$h_1$. Element~1 could not be routed to its bucket and, hence,
is greyed out (tagged $\false$ by the PRN). Hence, it is thrown into the next table.
In the next phases elements in $H_2$ and $H_3$ are routed according to~$h_2$ and~$h_3$, correspondingly. 
}
\label{fig:fullZHT}
\end{figure*}

{\section*{Appendix}}

\section{Additional lemmas}

\begin{lemma}
When throwing $m$ elements in $n$ buckets, each
with capacity $\BSize$, the probability that
a given bucket receives $\BSize$ or more elements
is at most ${m \choose \BSize} \frac{1}{n^{\BSize}}$.
\label{lemma:overflow_bucket_prob}
\end{lemma}
\begin{proof}
The probability that a bucket gets assigned $\BSize$ or more elements is
\begin{eqnarray*}
&&\sum_{i=c}^m  {m \choose i} \frac{1}{n^{i}}  \left(1 - \frac{1}{n}\right) ^{m - i} = \\
&&{m \choose c} \frac{1}{n^c} \sum_{i=c}^m  \frac{c! (m-c)!}{i! (m-i)!} \frac{1}{n^{i-c}}  \left(1 - \frac{1}{n}\right) ^ {m - i} = \\
&& {m \choose c} \frac{1}{n^{c}}\sum_{i=0}^{m-c}  \frac{c! (m-c)!}{(c+i)! (m-c-i!)} \frac{1}{n^i}  \left(1 - \frac{1}{n}\right) ^ {m - c - i}
\end{eqnarray*}
To prove the lemma, we show that the sum expression is at most 1.
Note that as long as the first component of the sum is at most ${m-c \choose i }$,
we can use the binomial formula:
$$\sum_{i=0}^{m-c}  {m -c \choose i}  \frac{1}{n^{i}} \left(1 - \frac{1}{n}\right) ^ {m - c - i} = 1$$
We now show that ${m-c \choose i } \ge \frac{c! (m-c)!}{(c+i)! (m-c-i!)}$ for $i \ge 0$
and $c \ge 1$:
\begin{eqnarray*}
\frac{(m-c)!}{i! (m-c-i)!}&\ge& \frac{c! (m-c)!}{(c+i)! (m-c-i)!}\\
\frac{1}{i! } &\ge& \frac{c! }{(c+i)!} = \frac{1}{(c+1)\ldots(c+i)} \\
\end{eqnarray*}
\end{proof}

\begin{lemma}
The average number of elements that do not get assigned to a bucket
during the throw of $m$ elements into $n$ buckets of size $\BSize$ is
at most $\frac{1}{n^\BSize} {m \choose \BSize+1}$.
\label{lemma:throw_ave}
\end{lemma}
\begin{proof}
Let $X_j$ be a random variable that is 1 if there is a collision on the $j$th inserted element,
and $0$~otherwise.
For the first $\BSize+1$ elements, $\Pr[X_1] = \Pr[X_2]  = \ldots = \Pr[X_{\BSize}] = 0$.
For $ c+1 \le j \le n$, we need to estimate the probability that $j$th element
is inserted in a bucket with at least $\BSize$ elements.
The probability of a given bucket having $c$ or more elements after $j-1$
elements have been inserted is given in Lemma~\ref{lemma:overflow_bucket_prob}.
Hence,
$$\Pr[X_{j}]  \le {j-1 \choose \BSize} \frac{1}{n^{\BSize}}$$
Note that $X_j$ are independent.
Then, expected number of collisions is
\begin{eqnarray*}
\mu = E\left[\sum_{j=1}^{m} X_j\right] \le \sum_{j={c+1}}^{m} \Pr[X_j] \le \\
\sum_{j={c}}^{m-1}  {j \choose \BSize} \frac{1}{n^\BSize} = \frac{1}{n^\BSize}  \sum_{j={c}}^{m-1}  {j \choose \BSize}
\end{eqnarray*}
Using the fact that $\sum_{j={c}}^{m-1}  {j \choose \BSize} = {m \choose\BSize+1}$:
\begin{eqnarray*}
\mu \le \frac{1}{n^\BSize} {m \choose \BSize+1}
\end{eqnarray*}
\end{proof}

\begin{lemma}
The number of elements that are spilled
during the $i$th stage of the routing network of size $n$ on input of size $m\le n$
is at most the number of elements that do not fit into $H_1$ of a ZHT
when throwing $m$ elements.
\label{lemma:network_indep_stage}
\end{lemma}
\begin{proof}
We are interested in counting the number of elements
whose tag changes from $\true$ to $\false$ during the~$i$th stage of
the probabilistic routing network and
refer to such elements as spilled during the $i$th stage.

Let $h_{start}$ be the hash function that was used to distribute elements
during the throw of $n$ elements into the input table of the network
(i.e., the elements that did find a spot in the table appear in buckets according to $h_{start}$)
and $h_{end}$ be the hash function of the network.
Let $||$ denote a concatenation of two bit strings.

Consider stage $i$ of the network, $1\le i \le \log n$.
After stage~$i$, element~$e$ is assigned to a bucket $a||b$ where~$a$ are the first $i$ bits of $h_{start}(e)$
and $b$ are the first $\log n -i$ bits of $h_{end}(e)$.
Now consider, a throw phase of $m$ elements into a table of $n$ buckets
using the same hash function.
Let us compare the number of elements assigned to a particular bucket
in each case. We argue that the number in the former case is either
the same as in the latter case or smaller.
Note that the elements that can be assigned to a bucket
is the same in both cases as elements are distributed across buckets using the same hash function.
However, some of the elements of the PRN
may have been lost in earlier stages (those that arrive to the $i$th stage with the~$\false$ tag)
and have already been counted towards the spill of an earlier stage.
Since elements tagged~$\false$ are disregarded by the network during routing,
$i$th stage may spill the same number of elements as the hash table
or less since there are less ``live'' (tagged $\true$) elements that reach
their buckets.
\end{proof}

\begin{lemma}
Consider the following allocation of $n$ elements
into
tables, each with $n$ buckets of size $c = \log \log n$.
The elements are allocated uniformly at random
in the first table. The elements that do not fit into the
first table are allocated into the second table.

The number of elements that arrive into the
second table is at most $n/(2e)$.
\label{lemma:firstounds}
\end{lemma}
\begin{proof}
In Lemma~\ref{lemma:throw_ave} we showed that average
number of outgoing elements at level $i$ is $\mu_i = \frac{1}{n^\BSize} {\inc_i \choose \BSize+1}$.
Hence, with $\inc_1=n$
$$\mu_1 \le \frac{n^{c+1}}{n^c(c+1)!} = \frac{n}{(c+1)!}\le \frac{n}{2 e\log n }$$
where we rely on $(c+1)! \ge 2e\log n$ (see below).
If $n/(2e \log n)$ is sufficiently large (e.g., $n/(2e\log n) \ge 3\secsize$
where $\secsize$ is the security parameter), we
can use Chernoff bounds to show that
with probability $\exp(-\mu_1/3)$ the number
of spilled elements from the first level is no more than $2\mu_1$.
Hence, in total the number of elements that spill is at most:
$$\frac{n}{e\log n } \le \frac{n}{2e}$$
If $n$ is of moderate size (e.g., $n \ge 2^{10}$), we can
use another form of Chernoff bounds~\cite[Chapter 4]{Mitzenmacher:2005:PCR:1076315}
to show that with probability $2^{-R}$
the number
of spilled elements from the first level is no more than $R\ge 6\mu_1$,
where $6\mu_1 \ge \secsize$ and $6\mu_1 \le n/2e$.

Finally we show that for any $a \ge 10$,
$(\log a +1)! \ge 2ea$:
\begin{eqnarray*}
(\log a +1)! \ge \sqrt{2 \pi (\log a+1)}  \frac{(\log a+1)^{\log a + 1}}{e^{\log a + 1}} \ge \\
 \sqrt{2 \pi (\log a+1)}  \frac{(\log a+1)^{\log a + 1}}{3^{\log a + 1}} \ge \\
 \sqrt{2 \pi (\log a+1)} \frac{2^{\log (\log a + 1)(\log a+1)}}{2^{(\log a + 1) \log 3}}\ge\\
 2ea
\end{eqnarray*}
\end{proof}

\begin{lemma}
Consider the process in Lemma~\ref{lemma:firstounds}
but after the first throw,
the first table is routed through the routing network
from Section~\ref{sec:network}
and elements spilled from the network
are also thrown into the second table.
The elements of the second table
are routed as well and spilled elements
are thrown into the third table, and so on.
The number of elements that arrive into the
third table is at most $n/(2e)$.
\label{lemma:firstounds2}
\end{lemma}
\begin{proof}
From Lemma~\ref{lemma:network_indep_stage}
we know that the total overflow from
throwing the elements and routing them through the network
is at most $\log n$ times the overflow from the original
throw.
Consider the proof of Lemma~\ref{lemma:firstounds}.
The second table receives at most
$$\frac{2 n \log n}{2e \log n} = \frac{n}{e}$$
number of elements.
Then the average overflow according to Lemma~\ref{lemma:throw_ave}
is
$$\frac{1}{n^c}{ \frac{n}{e} \choose c+1} \le \frac{n^{c+1}}{n^c e^{c+1} 2 e\log n} = \frac{n}{2\log n e^{c+2}}
\le \frac{n}{4 e \log n}$$
We can again use Chernoff bound to show that the number of elements that
arrive into the next level will be at most $2\log n$ factor away from the above expression.
Hence, the number of elements that arrive into the third table is
$$\frac{2 n \log n }{4 e \log n} = \frac{n}{2e}$$
\end{proof}

\begin{lemma}
Consider the following allocation of $m \le n/(2e)$ elements
into
$k$ tables, each with $n$ buckets of size $c = \log \log n$.
The elements are allocated uniformly at random
in the first table. The elements that do not fit into the
first table are allocated into the second table and so on,
until there are no elements left.
Then, $k=O(\log \log n)$ is sufficient to
allocate all elements.
\label{lemma:lastrounds}
\end{lemma}
\begin{proof}
We analyze the process using Poisson random variables
as an approximation of the balls-and-bins approximation process~\cite[Chapter 5]{Mitzenmacher:2005:PCR:1076315}.
Recall that as long as event's probability
is monotonically decreasing or increasing in the number of elements,
the event that happens with probability $\epsilon$ in the Poisson case has probability at most $2\epsilon$
in the exact case.
Let $Z_i$ be a Poisson random variable that measures
the number of elements in a bucket after a random throw.
Let $\alpha$ be the mean of $Z_i$.
Recall that,
$$\Pr(Z_i \ge c) \le \frac{e^{-\alpha} (e\alpha)^c}{c^c}$$
Hence, the number of $n$ buckets that will have
at least $c$ elements is $n {e^{-\alpha} (e\alpha)^c}/{c^c}$.
It is known, that with very high probability no bucket has more than $O(\log n)$
elements, hence, the number of overflow elements
is  $(n\log n) {e^{-\alpha} (e\alpha)^c}/{c^c}$.
We now show that as long as $\alpha \le 1/(2e)$, we require
$O(\log \log n)$ rounds until no overflow occurs.
After the first throw, the overflow is
\begin{eqnarray}
(n\log n) \frac{(e\alpha)^{c+1}}{(c+1)^{c+1}} \le 
\frac{n}{2^{c+1} e\log n} 
\label{eq:line1}
\end{eqnarray}
where with $c = \lceil\log \log n\rceil$,
$\log n /(c+1)^{c+1} \le 1/e$ (see below).
Hence, in the next round:
\begin{eqnarray*}
(n\log n) \frac{(e/{(2^{c+1}e\log n)})^{c+1}}{(c+1)^{c+1}} \le 
\frac{n}{2^{(c+1)^{2}} (\log n)^{c+2} e} 
\end{eqnarray*}
Then for the $i$th round the overflow is:
\begin{eqnarray*}
\frac{n}{2^{(c+1)^{i}} (\log n)^{(c+1)^{i-1}} e} 
\end{eqnarray*}
Hence, after at most $O(\log \log n)$ rounds
the above expression is less than 1.

We are left to show that as long as $a \ge 8$,
$a/(\log a + 1)^{(\log a +1)} \le 1/e$.
In the following lemma
we require a similar result but for~$a^2$
in the numerator. Hence, instead we show that
$a^2/(\log a + 1)^{(\log a +1)} \le 1/e$:
\begin{eqnarray*}
\frac{a^2}{(\log a + 1)^{(\log a +1)}} =
\frac{1}{2^{\log (\log a+1) (\log a + 1) - 2\log a}} \le\\
\frac{1}{2^{ \log (\log a+1)  + \log a (\log (\log a+1)  -2)}} \le 
\frac{1}{\log a + 1} \le \frac{1}{4} \le \frac{1}{e}
\end{eqnarray*}
\end{proof}

\begin{lemma}
Consider the process in Lemma~\ref{lemma:lastrounds}
but after the first throw,
the first table is routed through the routing network
from Section~\ref{sec:network}
and elements spilled from the network
are also thrown into the second table.
The elements of the second table
are routed as well and spilled elements
are thrown into the third table, and so on.
Then, $O(\log \log n)$ rounds
are still sufficient to avoid overflow
in the last round.
\label{lemma:lastrounds2}
\end{lemma}
\begin{proof}
From Lemma~\ref{lemma:network_indep_stage}
we know that the total overflow from
throwing the elements and routing them through the network
is at most $\log n$ times the overflow from the original
throw.
Consider the proof of Lemma~\ref{lemma:lastrounds}.
In Equation~\ref{eq:line1},
we can replace $n$ with $n \log n$ to account
for the buckets in each round of the network.
This would result in the spill
of
$$
\frac{n}{2^{c+1} e\log n} 
$$
in the first round
and
\begin{eqnarray*}
\frac{n}{2^{(c+1)^{i}}  e} 
\end{eqnarray*}
in the $i$th round.
Hence, $O(\log \log n)$ rounds are still sufficient
to fit all the elements.
\end{proof}

\end{document}